\documentclass[12pt,a4paper]{amsart}
\usepackage[utf8]{inputenc}
\usepackage[english]{babel}
\usepackage{lmodern}
\usepackage{array,booktabs}
\usepackage{multicol}
\usepackage{mathrsfs, amssymb, amsmath, amsfonts}
\usepackage{graphicx}
\usepackage{geometry, xcolor}
\usepackage{enumitem}
\usepackage[abs]{overpic}
\usepackage[pagewise]{lineno} % fuer Zeilennummerierung
\allowdisplaybreaks[2]

\title[A Non-Local Quasi-Linear Ground State Representation]{A Non-Local Quasi-Linear Ground State Representation and Criticality Theory}

\author{Florian Fischer}
\address{Florian Fischer, Institute of Mathematics, University of Potsdam, Karl-Liebknecht-Straße 24-25, 14476 Potsdam, Germany}
\email{florifis@uni-potsdam.de}
\usepackage[style=alphabetic,maxnames=5,minnames=2,abbreviate=false,date=long,url=false,isbn=false, doi=false,backend=biber, backref,backrefstyle=none]{biblatex}
\usepackage{csquotes}
\AtBeginBibliography{\tiny}
\usepackage[colorlinks=true,citecolor=blue,linkcolor=blue]{hyperref}

\newtheorem{theorem}{Theorem}[section]
\newtheorem{lemma}[theorem]{Lemma}
\newtheorem{proposition}[theorem]{Proposition}
\newtheorem{corollary}[theorem]{Corollary}
\newtheorem{ass}[theorem]{Assumption}

\theoremstyle{definition}
\newtheorem{example}[theorem]{Example}
\newtheorem{remark}[theorem]{Remark}

\numberwithin{equation}{section}

\newcommand{\abs}[1]{\left\lvert #1\right\rvert} %%Command: \abs{X} gives  |X|
\newcommand{\set}[1]{\left\{ #1\right\} }
\newcommand\ip[2]{\langle #1, #2 \rangle}
\newcommand{\p}[1]{\left( #1 \right)^{\langle p-1 \rangle}}
\newcommand{\sse}{\subseteq}

\renewcommand{\epsilon}{\varepsilon}
\renewcommand{\phi}{\varphi}
\newcommand{\NN}{\mathbb{N}}

\newcommand{\RR}{\mathbb{R}}

\newcommand{\dd}{\mathrm{d}}

\DeclareMathOperator{\sgn}{sgn}

\DeclareMathOperator{\cc}{cap}

\DeclareMathOperator{\Div}{div}
\newcommand{\FF}{F}%{\tilde{F}}
\newcommand{\DD}{D}%{\tilde{D}}

\newcommand{\bb}{\tilde{b}}
\newcommand{\hh}{\tilde{h}}
\newcommand{\HH}{\tilde{H}}
\newcommand{\ccc}{\tilde{c}}

\newcommand{\uu}{\tilde{u}}

\newcommand{\Hmm}[1]{\leavevmode{\marginpar{\tiny%
			$\hbox to 0mm{\hspace*{-0.5mm}$\leftarrow$\hss}%
			\vcenter{\vrule depth 0.1mm height 0.1mm width \the\marginparwidth}%
			\hbox to 0mm{\hss$\rightarrow$\hspace*{-0.5mm}}$\\\relax\raggedright #1}}}
\begin{document}

\begin{abstract}
We study energy functionals associated with non-local quasi-linear Schrödinger operators, and develop a  ground state representation. Our  main focus is on infinite graphs but we also consider non-local quasi-linear Schrödinger operators in the Euclidean space. Using the representation, we develop a criticality theory for quasi-linear Schrödinger operators on general weighted graphs, and show characterisations for a Hardy inequality to hold true. As an application, we show a Liouville comparison principle on graphs. %We also show briefly a ground state representation for weighted fractional $p$-Laplacians in the Euclidean space.
\\
	\\[4mm]
	\noindent  2020  \! {\em Mathematics  Subject  Classification.}
	Primary  \! 39A12; Secondary  31C20; 31C45; 35R02.
	\\[4mm]
\end{abstract}

\maketitle

\section{Introduction}
The quasi-linear $p$-Laplacian on Riemannian manifolds, and especially in the Euclidean space, is one of the best studied non-linear local elliptic operators. There are many beautiful monographs related to this operator, see e.g. \cite{Appell, BEL15, DKN97, HKM, Lindqvist}, or in general metric spaces, see \cite{Bjoern}.

%In this paper, we consider the non-local analogue and generalisations on graphs and on $\RR^{d}$. %We show our main result in detail on graphs and briefly on $\RR^{d}$.

First results for $p$-Laplacians on finite graphs are given in \cite{AmghibechPicone, CL11, HoS, HoS2, ParkChung, ParkKimChung}. On locally finite graphs, $p$-Laplacians were  studied in \cite{Prado}, and on almost locally finite graphs in \cite{Mu13, SY93Class}.

%In the present paper we deal with an extended class of operators, which include, as a particular case, the $p$-Laplacian.
Here, we study an extended class of operators including the $p$-Laplacian as a special case, i.e., $p$-Schrödinger operators. Moreover, we study this class on  locally summable weighted graphs.

Furthermore, we extend our main result also to non-local $p$-Schrödinger operators on $\RR^{d}$. Here, $p$-Laplace-type operators were studied in \cite{BF14,CMS18, DCKP16,FS08, Kuusi15}.

Recently, the potential theory of local $p$-Schrödinger operators with not necessarily non-negative potential term was studied more closely, see e.g. \cite{BEL15, DP16, HPR21, PP, PR15, PT, PTT08}.  In this theory, the ground state representation is of fundamental importance. This representation is an equivalence between functionals. It states that the $p$-energy functional associated with the $p$-Schrödinger operator is equivalent to a simplified energy functional consisting of non-negative terms only.

For non-local $p$-Schrödinger operators in the Euclidean space, which includes graphs as a special case, a one-sided inequality for $p \geq 2$ is given in \cite{FS08}. 

Here, we show a ground state representation for non-local $p$-Schrödinger operators for all $p>1$ in terms of an equivalence between the corresponding $p$-energy functional and the simplified energy. We show this statement on graphs, see Theorem~\ref{thm:GSR} and Corollary~\ref{cor:GSR}. Moreover, we briefly show the corresponding result for non-local $p$-Schrödinger operators on $\RR^{d}$, see Theorem~\ref{thm:fracGSR}, and get as a consequence an improvement of a result in \cite{FS08}. 

The ground state representation is an essential tool in criticality theory (which is sometimes also called  parabolic theory). A $p$-energy functional is called critical if it is non-negative and the $p$-Hardy inequality does not hold. In the continuum, there are many characterisations of criticality known, see e.g. \cite{HPR21, PP}. 

Criticality theory on almost locally finite graphs for the standard $p$-Laplacian was studied in \cite{SY93Class}, see \cite{Prado} for locally finite graphs. Both show a connection between positive harmonic functions and the variational $p$-capacity.

In this paper, we establish a characterisation of criticality for $p$-Schrödinger operators in terms of null-sequences, the variational $p$-capacity, as well as positive harmonic functions, see Theorem~\ref{thm:critical}. This will be achieved mainly by the aid of the ground state representation. These characterisations are the discrete  counterpart to results in \cite{HPR21, PP}.

Moreover, we show a Liouville comparison principle which is a discrete analogue to results in \cite{PR15, PTT08}. It is another application of the ground state representation and gives the criticality of an energy functional if a subharmonic function can be estimated properly by the ground state of another energy functional.

There is an abundance of literature for the linear $(p=2)$-case: For linear Schrödinger operators on locally summable graphs see \cite{HK, KLW21, KePiPo2, KePiPo1}, for random walks see \cite{Woe-Book}, for Schrödinger forms see \cite{Fuk,Pinsky95,Tak14, TU21}, for Jacobi matrices see \cite{FSW08}, and references therein. In probabilistic settings, the ground state representation is also known as Doob's $h$-transform.

%Let us briefly discuss the linear case (that is $p=2$): On   locally summable graphs, the ground state representation associated with $2$-Schrödinger operators is an equality. It was established in \cite{KePiPo2, KePiPo1}, where it was used to build up a profound criticality theory. In a special setting the representation is also known as Doob's $h$-transform, and criticality is also known as recurrence. This notion appears in the setting of random walks, positive matrices and Dirichlet forms, see \cite[Remark~5.8]{KePiPo1} and the monographs \cite{Fuk,Woe-Book}. %Moreover, in this context, the ground state representation 

In the local and linear case, ground state representations are classical and have shown their powerfulness in many applications, see \cite[Section~1]{FS08} for a list of applications with references and more details. 

%For general graphs in the setting of this article, we do not know of any previous paper for the quasi-linear potential theory of the $p$-Laplacian. For the case of the $p$-Laplacian with potential term, we even do not know of any paper about quasi-linear potential theory with respect to locally finite graphs. %%% Takeuchi !!!

%In this paper, we will establish a characterisation of criticality in terms of the existence of null-sequences, as well as in the existence of a unique positive harmonic function. This will be achieved mainly by the aid of a ground state representation.

This paper is organised as follows: In Section~\ref{sec:setting}, we briefly introduce the basic notation and show connection between the $p$-Schrödinger operators and $p$-energy functionals via a Green's formula. Then, we turn in Section~\ref{sec:ground} to the main results of this paper, Theorem~\ref{thm:GSR} and Corollary~\ref{cor:GSR}. After we stated the results, we discuss them  in detail in Subsection~\ref{sec:discuss}. This includes a comparison with the local case in the continuum. The proofs of Theorem~\ref{thm:GSR} and Corollary~\ref{cor:GSR} are then divided into two parts: The proof of an elementary equivalence in Subsection~\ref{sec:elementary}, and then the application of this equivalence in Subsection~\ref{sec:proofGSR}. Thereafter, we show in Section~\ref{sec:critical} how this ground state representation can be used to prove some characterisations of criticality. Here, also one part of an Agmon-Allegretto-Piepenbrink-type theorem is needed, as well as a local $p$-Harnack inequality. This is part of Subsection~\ref{sec:preliminaries}. We end this paper with two simple applications of Theorem~\ref{thm:GSR} and Theorem~\ref{thm:critical}, one of them is a Liouville comparison principle.% which is a discrete analogue to results in \cite{PR15, PTT08}.

\section{Setting the Scene}\label{sec:setting}%\Hmm{betrachten eigl nur formen auf Cc. Sollte so auch im Text sichtbar werden}

In this section, we start by introducing graphs. Thereafter, we define quasi-linear Schrödinger operators on graphs. We end this section by introducing $p$-energy functionals and showing a connection to $p$-Schrödinger operators via Green's formula.

%We extend results from \cite{KePiPo1} to $p\neq 2$, and results from \cite{PP} to graphs, and also both from Schrödinger operators to Schrödinger-type operators. %We are not aware of similar results for our setting. 
% For an introduction to the standard $p$-Laplacian on locally finite graphs, see \cite{Prado}, and for $p$-Schrödinger operators on finite graphs, see \cite{ParkKimChung}. %However, the here presented setting is more general in both the operators and the graphs. %and has a strong focus on energy functionals associated with  $p$-Schrödinger operators.

%{\color{red}There exists some somehow similar results for $p$-Schrödinger operators on finite graphs, see \cite{ParkKimChung} and references therein, and for the standard $p$-Laplacian on infinite graphs, see \cite{Prado}.} 

%\Hmm{Wenn man will, dass der Sobolev Raum ein Banach Raum ist, dann muss $c\geq 0$ gelten. Das hat aber zur Folge, dass $h\geq 0$ auf ganz $\DD$ ist. Folglich macht ein AP Theorem nicht mehr sonderlich viel Sinn. Folglich sollte ich nicht $c\geq 0$ annehmen und werde somit nicht in Banach Räumen agieren.}
\subsection{Graphs and Schrödinger Operators}%\Hmm{stark kürzen und auf das Wesentliche begrenzen}
Let an infinite set $X$ equipped with the discrete topology and a symmetric function  $b\colon X\times X \to [0,\infty)$ with zero diagonal be given such that $ b $ is locally summable, i.e., the vertex degree satisfies \[ \deg(x)=\sum_{y\in X}b(x,y)<\infty, \qquad x\in X.\]  
We refer to $ b $ as a \emph{graph} over $X$ and elements of $X$ are called \emph{vertices}. Two vertices $x, y$ are called \emph{connected} with respect to the graph $b$ if $b(x,y)>0$, in terms $x\sim y$. A subset $V\sse X$ is called \emph{connected} with respect to $b$, if for every two vertices $x,y\in V$ there is a vertices ${x_0,\ldots ,x_n \in V}$, such that $x=x_0$, $y=x_n$ and $x_{i-1}\sim x_i$ for all $i\in\set{1,\ldots, n-1}$. For $V\sse X$ let $\partial V=\set{y\in X\setminus V\colon y\sim z\in V}$.
Throughout this paper we will always assume that 
\begin{center}$X$ is connected with respect to the graph  $b$.\end{center}

We now turn to functions: Let $S$ be some arbitrary set. A function $f\colon S\to \RR$ is called \emph{non-negative}, \emph{positive}, or \emph{strictly positive} on $I\sse S$, if $f\geq 0$, $f\gneq 0$, $f>0$ on $I$, respectively. If for two non-negative functions $f_1, f_2\colon S\to \RR $ there exists a constant $C>0$ such that $C^{-1}f_1 \leq f_2 \leq Cf_1$ on $I\subset S$, we write \[f_1 \asymp f_2 \quad\text{ on }I,\]
and call them \emph{equivalent} on $I$. %A function $f\colon \RR\to \RR$ is \emph{odd}, if $f(t)=-f(-t)$, $t\in \RR$.

The space of real valued functions on $V\subseteq X$ is denoted by $C(V)$ and is a subspace of $C(X)$ by extending the functions of $C(V)$ by zero on $X\setminus V$. The space of functions with compact support in $V$ is denoted by $ C_c(V)$. 

A strictly positive function $m\in C(X)$ extends to a measure with full support via ${m(V)= \sum_{x\in V}m(x)}$ for $V\sse X$.

The next fundamental definition is the one of the $p$-Laplacian. But first, we have to introduce some notation. For showing the connection to the counterpart in the continuum, we introduce the difference operator $\nabla$ on $C(X)$ via
\[\nabla_{x,y}f=f(x)-f(y), \qquad x,y\in X.\]

Let $p\in[1,\infty)$. For $V\sse X$, let the \emph{formal space} $ \FF(V)=\FF_{b,p}(V) $ be given by
\begin{align*}
\FF(V)= \{ f\in C(X): \sum_{y\in X} b(x,y)\abs{\nabla_{x,y}f}^{p-1} < \infty  \mbox{ for all } x\in V  \}.
\end{align*}

If $V=X$ we write $\FF=\FF(X)$. 

For $1<p<2$ we make the convention that $\abs{t}^{p-2}t=0$ if  $t=0$, i.e., $0\cdot \infty=0$. Then, we can write for all $p\geq 1$, 
\[ \p{t}:= |t|^{p-1} \sgn (t)=|t|^{p-2} t, \qquad t\in \RR. \]
Here, $\sgn\colon \RR\to \set{-1,0,1}$ is the sign function, that is $\sgn(t)=1$ for all $t > 0$, $\sgn(t)=-1$ for all $t< 0$, and $\sgn(0)=0$. We remark that $\FF(V)=C(X)$ if $p=1$, by the local summability assumption on the graph.

Next, we show a basic lemma, which states an alternative representation for the formal space. There, we need the following elementary inequality: we have for all $p\geq 0$ that
\begin{align}\label{eq:pTriangle}
\abs{\alpha+\beta}^{p}\leq 2^{p}(\abs{\alpha}^{p}+\abs{\beta}^{p}), \qquad \alpha,\beta\in \RR.	
\end{align}
This follows from $\abs{\alpha +\beta}^{p}\leq (2\max\set{\abs{\alpha},\abs{\beta}})^{p}\leq 2^{p}(\abs{\alpha}^{p}+\abs{\beta}^{p})$.

\begin{lemma}\label{lem:Fwelldefined}
	Let $V\sse X$ and $p\geq 1$. Then,
	\[\FF(V)= \{ f\in C(X): \sum_{y\in X} b(x,y)\abs{f(y)}^{p-1} < \infty  \mbox{ for all } x\in V  \}. \]
	In particular, $C_c(X)\sse \ell^{\infty}(V)\sse \FF(V)$.
\end{lemma}
\begin{proof}
	The case $p=1$ is trivial. Let $p>1$, and denote the set on the right-hand side by $\hat{F}(V)$. We obviously have that $C_c(V)\sse\ell^\infty(V)\sse \hat{F}(V)$. Furthermore, let $f\in \hat{F}(V)$. Then, using the elementary inequality \eqref{eq:pTriangle},  we get for any $x\in V$  that
	\begin{align*}
		\sum_{y \in X} b(x,y)\abs{\nabla_{x,y}f}^{p-1}
		\leq 2^{p-1}\Bigl( \abs{f(x)}^{p-1}\sum_{y\in X} b(x,y)+ \sum_{y\in X} b(x,y)\abs{f(y)}^{p-1}\Bigr).
	\end{align*}
	The first sum on the right-hand side is finite by the local summability property of the graph $b$. The second sum is finite since $f\in \hat{F}(V)$. This shows $f\in \FF(V)$.
	
	Moreover, if $f\in \FF(V)$ we obtain $f\in \hat{F}(V)$ since for all $x\in V$
		\begin{align*}
		\sum_{y \in X} b(x,y)\abs{f(y)}^{p-1}
		\leq 2^{p-1}\Bigl( \abs{f(x)}^{p-1}\sum_{y\in X} b(x,y)+ \sum_{y\in X} b(x,y)\abs{\nabla_{x,y}f}^{p-1}\Bigr)< \infty.\qquad \qedhere
	\end{align*}
\end{proof}
Now, we are in a position to define the Laplacian: Let $m$ be a measure on $X$.  
%Moreover, let $B_{p}\colon \RR\to \RR$ be a skew-symmetric function which is equivalent to $\abs{\cdot}^{p-1}$ on $[0,\infty)$, that is 
%\begin{align}\label{Bequiv}
%	B_{p}(t)=-B_{p}(-t),\,\, t\in \RR,\qquad  \text{and} \quad B_{p}\asymp \abs{\cdot }^{p-1},\quad \text{on }[0,\infty).
%\end{align}
%The second condition can also be rewritten as $B_{p}(t)t\asymp \abs{t}^{p}$, $t\in \RR$. The assumptions on $B_p$ ensure that the resulting Schrödinger operator and energy functional can be associated on $C_c(X)$. Moreover, for $p>1$ we get $B_p(0)=0$. 
Then, the \emph{($p$-)Laplace operator} $L=L_{b, m, V,p} \colon \FF(V)\to C(V)$ is defined via
%\[ Lf(x)=\frac{1}{m(x)}\, \sum_{y\in X} b(x,y)\abs{\nabla_{x,y}f}^{p-1}\sgn(\nabla_{x,y}f), \qquad x\in V.\]
\begin{align*}
	Lf(x)&=\frac{1}{m(x)} \sum_{y\in X} b(x,y)\p{\nabla_{x,y}f},\qquad x\in V.
\end{align*}

Let $p\geq 1$. If we have additionally $m=1$, $b(X\times X)\sse \set{0,1}$, then $L$ is called \emph{standard $p$-Laplacian}. %Note that the operator does not see loops.

\begin{remark}
	Following \cite{Mu13, Prado, Takeuchi}, there is the following analogy to $p$-Laplacians in the continuum:  A vector field $v$ is a function in $C(X\times X)$ such that $v(x,y)=-v(y,x)$, $x,y\in X$. Moreover, define $\Div$ on the space of absolutely summable vector fields in the second entry via
	\[(\Div v)(x)=\frac{1}{m(x)}\sum_{y\in X}v(x,y).\]
	Then, for all $f\in \FF$, and $p\geq 1$,
	\[Lf(x)= \Div (b\abs{\nabla f}^{p-2}\nabla f)(x), \qquad x\in X.\]
	This shows that our Laplacian is a discrete analogue to weighted Laplace-Beltrami-type operators on manifolds.
\end{remark}

Finally, we can define Schrödinger operators as follows: Let $c\in C(X)$. %, and $C\colon \RR\to \RR$.
Then the \emph{($p$-)Schrödinger operator} $H=H_{b,c,m,V,p}\colon \FF(V)\to C(V)$ is given by
\[Hf(x)=Lf(x)+\frac{c(x)}{m(x)}\p{f(x)}, \qquad x\in V.\]

The function $c$ is then usually called the \emph{potential} of $H$. If $c$ is non-negative, then $H$ is called \emph{$p$-Laplace-type} operator.

A function $u\in \FF(V)$ is said to be \emph{harmonic, (superharmonic, subharmonic)} on $V\sse X$ with respect to $H$ if \[Hu=0 \quad(Hu\ge 0,\, Hu\leq 0)\qquad\text{ on }V.\] 
If $V=X$ we only speak of super-/sub-/harmonic functions.
%\Hmm{kürzen?}A function $u\in \FF(V)$ is said to be a \emph{(p-)solution, ((p-)supersolution, (p-)sub-solution)} on $V\sse X$ with respect to $H$ and $g\in C(V)$ if \[Hu=g \quad(Hu\ge g,\, Hu\leq g)\qquad\text{ on }V.\] 
%If $g=0$ we speak of \emph{(p-)harmonic, ((p-)superharmonic, (p-)subharmonic)} functions on $V$. If $V=X$ we only speak of super-/sub-/harmonic functions, respectively super-/sub-/solutions with respect to $g$.

\subsection{Energy Functionals Associated with Graphs}%\Hmm{nur in Cc definieren? Dann kann man es auch nach ganz vorne packen...}
Let $\DD=\DD_{b,c, p}$ be given by
\begin{align*}
\DD=\bigl\{f\in C(X): \sum_{x,y\in X} b(x,y)\abs{\nabla_{x,y}f}^{p}+ \sum_{x\in X}\abs{c(x)}\abs{f(x)}^p<\infty \bigr\}.%\\
%&=\bigl\{f\in \ell^p(X,\abs{c})\colon \sum_{x,y\in X} b(x,y)\abs{\nabla_{x,y}f}^p<\infty \bigr\}.
\end{align*}
%\Hmm{Hier muss $|c|$ gewählt werden, das es sonst auch gegen $-\infty$ gehen kann! Oder es soll dann nur ein Wert in $\RR$ sein.}
Then, the \emph{($p$-)energy functional} $h=h_{b,c,p}\colon \DD\to \RR$ %to the Schrödinger-type operator $H$ 
is defined via 
\[ h(f)=\frac{1}{2}\sum_{x,y\in X} b(x,y)\abs{\nabla_{x,y}f}^{p}+ \sum_{x\in X}c(x)\abs{f(x)}^{p}.\]

%We call the energy functional \emph{classical} if $C=B=\abs{\cdot}^{p-1}\sgn$, $p\geq 1$. 
If $p=2$, then the energy functional is a quadratic form, and called Schrödinger form.

As in the continuum or the linear case on graphs, there exists a so-called Green's formula which shows a connection between $H$ and $h$ on $C_c(X)$. The Green's formula seems to be folklore in both worlds. However, for the convenience of the reader we include a proof here. A similar proof of the Green's formula for the normalised $p$-Laplacian, that is $m=\deg$ and $c=0$,  is given in \cite{Takeuchi}. %Here we need in the calculations that $B$ is odd, and to ensure that we have absolute convergent sums and $HC_c(X)\sse C(X)$, we use additionally that $B$ is  equivalent to $\abs{\cdot}^{p-1}$ on $[0,\infty)$.
\begin{lemma}[Green's formula]\label{lem:GreensFormula}
	Let $p\geq 1$, $V\sse X$, $f\in \FF(V)$ and $\phi\in C_c(X)$. Then, all of the following sums converge absolutely and
	\begin{multline*}
\sum_{x\in V}Hf(x)\phi(x)m(x)
=\frac{1}{2}\sum_{x,y\in V} b(x,y)\p{\nabla_{x,y}f}( \nabla_{x,y}\phi)+  \sum_{x\in V}c(x)\p{f(x)}\phi(x)\\
+\sum_{x\in V, y\in \partial V}b(x,y)\p{\nabla_{x,y}f}\phi(x). 
	\end{multline*}
%	\begin{align*}
%		\sum_{x\in V}Hf(x)\phi(x)m(x)
%		&=\frac{1}{2}\sum_{x,y\in V} b(x,y)\abs{\nabla_{x,y}f}^{p-2}(\nabla_{x,y}f)( \nabla_{x,y}\phi) \\
%		&\qquad+ \sum_{x\in V}c(x)\abs{f(x)}^{p-2}f(x)\phi(x)\\
%		&\qquad +\sum_{x\in V, y\in \partial V}b(x,y)\abs{\nabla_{x,y}f}^{p-2}(\nabla_{x,y}f)\phi(x). \\
%	\end{align*}
	In particular, the formula can be applied to $f\in C_c(X)$, or $f\in \DD$, and
	\[h(\phi)=\sum_{x\in V}H\phi(x)\phi(x)m(x), \qquad \phi\in C_c(V).\]
	%All equivalences become equalities if $B_{p}=C_{p}= \abs{\cdot }^{p-1}\sgn(\cdot).$
\end{lemma}
\begin{proof}
	Since $\phi\in C_c(X)$, the absolute convergence follows from 
	\begin{align*}
	\sum_{x\in V}\abs{Lf(x)\phi(x)}m(x)\leq \sum_{x\in V}\abs{\phi(x)}\sum_{y\in X}b(x,y)\abs{\nabla_{x,y}f}^{p-1}< \infty,
	\end{align*}
	for any $f\in \FF(V)$.
	
	Applying Fubini's theorem,  using the absolute convergence of the sums and the symmetry of $b$, we get
	\begin{multline*}
\sum_{x\in V}Lf(x)\phi(x)m(x) = \sum_{x\in V, y\in X}b(x,y)\p{\nabla_{x,y}f}\phi(x) \\
=\frac{1}{2}\sum_{x,y\in V} b(x,y)\p{\nabla_{x,y}f}\phi(x) -\frac{1}{2}\sum_{\hat{x},\hat{y}\in V}b(\hat{x},\hat{y})\p{\nabla_{\hat{x},\hat{y}}f}\phi(\hat{y})\\
+\sum_{x\in V, y\in \partial V}b(x,y)\p{\nabla_{x,y}f}\phi(x)\\
=\frac{1}{2}\sum_{x,y\in V} b(x,y)\p{\nabla_{x,y}f}\nabla_{x,y}\phi +\sum_{x\in V, y\in \partial V}b(x,y)\p{\nabla_{x,y}f}\phi(x).
	\end{multline*}
%	\begin{align*}
%		\sum_{x\in V}Lf(x)\phi(x)m(x) &= \sum_{x\in V, y\in X}b(x,y)\bigl(\nabla_{x,y}f\bigr)\abs{\nabla_{x,y}f}^{p-2}\phi(x) \\
%		&=\frac{1}{2}\sum_{x,y\in V} b(x,y)\bigl(\nabla_{x,y}f\bigr)\abs{\nabla_{x,y}f}^{p-2}\phi(x) \\
%		&\qquad -\frac{1}{2}\sum_{\hat{x},\hat{y}\in V}b(\hat{x},\hat{y})\bigl(\nabla_{\hat{x},\hat{y}}f\bigr)\abs{\nabla_{\hat{x},\hat{y}}f}^{p-2}\phi(\hat{y}) \\
%		&\qquad +\sum_{x\in V, y\in \partial V}b(x,y)\bigl(\nabla_{x,y}f\bigr)\abs{\nabla_{x,y}f}^{p-2}\phi(x) \\
%		 &=\frac{1}{2}\sum_{x,y\in V} b(x,y)\bigl(\nabla_{x,y}f\bigr)\abs{\nabla_{x,y}f}^{p-2}\nabla_{x,y}\phi \\
%		 &\qquad +\sum_{x\in V, y\in \partial V}b(x,y)\bigl(\nabla_{x,y}f\bigr)\abs{\nabla_{x,y}f}^{p-2}\phi(x). \\
%	\end{align*}
	The assertions for the Schrödinger operator $H$ follow now easily. 
	
	By Lemma~\ref{lem:Fwelldefined}, $C_c(X)\sse \FF(V)$.  Note that $\FF(X)\sse \FF(V)$. It remains to show that $\DD\sse \FF(X)$. This follows from Hölder's inequality for all $x\in X$ by
	\begin{align*}
		\sum_{y\in X}b(x,y)\abs{\nabla_{x,y}f}^{p-1}\leq \Bigl( \sum_{y\in X}b(x,y) \Bigr)^{1/p}\Bigl( \sum_{y\in X}b(x,y)\abs{\nabla_{x,y}f}^{p} \Bigr)^{(p-1)/p}< \infty.
	\end{align*}
	This ends the proof.
\end{proof}
%%\begin{remark}
%%Another connection between $h$ and $H$ is that $p\cdot H$ is the G\^{a}teaux derivative of $h$ on $C_c(X)$, which can be seen as follows: for all $\phi, \psi\in C_c(X)$ we have
%%	\begin{align*}
%%		\frac{\dd}{\dd\, t}\Bigl.  h(\phi+t\, \psi) \Bigr|_{t=0} %&=\frac{\dd}{\dd t}\Bigl. \frac{p}{2}\sum_{x,y\in X}b(x,y)\abs{\nabla_{x,y}\phi+t(\nabla_{x,y}\psi)}^{p-1}\sgn \bigl(\nabla_{x,y}\phi+t(\nabla_{x,y}\psi)\bigr)\bigl(\nabla_{x,y}\psi\bigr) \Bigr|_{t=0} \\
%%		&=\frac{p}{2}\sum_{x,y\in X}b(x,y)\abs{\nabla_{x,y}\phi}^{p-1}\sgn \bigl(\nabla_{x,y}\phi\bigr)\bigl(\nabla_{x,y}\psi\bigr) \\
%%		&\quad+ p\sum_{x\in X}c(x)\abs{\phi(x)}^{p-1}\sgn(\phi(x))\psi(x)\\
%%		&=p \sum_{x\in X}H\phi(x)\psi(x)m(x).
%%	\end{align*}
%%\end{remark}
%
%We say that $h$ is \emph{non-negative} on $D\sse \DD$, if $h(\phi)\geq 0$ for all $\phi\in D$.
%
%If for two energy functionals $h_1$ and $h_2$, defined on $\DD_1$ and $\DD_2$, respectively, there exists a constant $C>0$ such that $C^{-1}h_1 \leq h_2 \leq Ch_1$ on $D\sse \DD_1\cap \DD_2$, we write \[h_1 \asymp h_2 \quad\text{ on }D,\]
%and call them \emph{equivalent} on $D$. 
%
%{\color{red}A non-negative function $w\in C(X)$ gives rise to a canonical ($p$-)functional $w=w_p\colon C_c(X)\to [0,\infty)$ via
%\[w(\phi)=\sum_{x\in X}\abs{\phi(x)}^{p}w(x).\]
%We will change between the function $w$ and the corresponding form $w$ without further notice (as long as it is clear which representation we take).}
%
%
\section{The Ground State Representation on Graphs}\label{sec:ground}
%linear case in die Einleitung als Motivation?
In the classical linear case, ground state representations are transformations which use a superharmonic function	to turn a quadratic energy form associated with a linear Schrödinger operator into a quadratic energy form associated with a linear Laplace operator, see e.g. \cite[Proposition~4.8]{KePiPo1} for such a statement on graphs, and e.g. \cite[p. 109]{DavHeat} for a counterpart in the continuum.

In the non-linear ($p\neq 2$)-case, we do not have an equality via a transformation between functionals anymore. But instead, we achieve an equivalence between functionals, providing that a positive superharmonic function exists. The equivalent functional has the property that it consists of non-negative terms only. 

Our representations in Theorem~\ref{thm:GSR} and Corollary~\ref{cor:GSR} can be seen as the non-local analogues to the local and non-linear representations in \cite{PR15, PTT08}, where   $p$-Schrödinger operators on domains in $\RR^d$ are discussed. Furthermore, we briefly show the representation for weighted non-local $p$-Schrödinger operators in $\RR^{d}$ in the next section, Section~\ref{sec:frac}.%, and improve a result in \cite{FS08} significantly.  %Loosely speaking it says that the form to the Schrödinger operator can be approximated from below and above by a form associated with a Schrödinger operator with non-negative potential part (up to constants). %For $p=2$, the equivalence is indeed an equality and equals the one in \cite{KePiPo1}.  %The corresponding analogue in the continuum for the $p$-Laplacian can be found in \cite{PTT08} and is there called \emph{simplified energy}.

First applications of our representations are given in Section~\ref{sec:critical}. Moreover, other applications can be found in the follow-up papers \cite{F:Opti, F:AAP}. 

\subsection{The Statement}\label{sec:GSRepStatement}
Let $p>1$, and $0\leq u\in \FF(V)$ for some $V\sse X$. The \emph{simplified energy (functional)} $h_{u}$ of $h$ with respect to $u$ on $C_c(V)$ be given by
%\[h_u(\phi):=\sum_{x,y\in X}b(x,y) u(x)u(y)(\nabla_{x,y}\phi)^{2}\bigl( \abs{\nabla_{x,y}u\phi}+\abs{\phi_{u}(x,y)}\abs{\nabla_{x,y}u} \bigr)^{p-2},\]
\begin{align*}
%h_u(\phi)&:=\sum_{x,y\in X}b(x,y) u(x)u(y)(\nabla_{x,y}\phi)^{2}\bigl( %\abs{\nabla_{x,y}u\phi}+\abs{\phi_{u}(x,y)}\abs{\nabla_{x,y}u} \bigr)^{p-2},\\
	h_{u}(\phi)&:=\sum_{x,y\in X}b(x,y) u(x)u(y)(\nabla_{x,y}\phi)^{2} \\
		&\qquad \cdot\left( \bigl(u(x)u(y)\bigr)^{1/2}\abs{\nabla_{x,y}\phi}+ \frac{\abs{\phi(x)}+ \abs{\phi(y)}}{2}\abs{\nabla_{x,y}u} \right)^{p-2},
\end{align*}
where we set $0\cdot \infty =0$ if $1<p<2$. % and $(u(x)u(y))^{1/2}\abs{\nabla_{x,y}\phi}+ (\abs{\phi(x)}+ \abs{\phi(y)})\abs{\nabla_{x,y}u}/2=0$.

Moreover, we define a weighted bracket $\ip{\cdot}{\cdot}$ on $C(X)\times C_c(X)$ via
\[\ip{f}{\phi}:=\sum_{x\in X}f(x)\phi(x)m(x),\qquad f\in C(X), \phi\in C_c(X).\]
%and the \emph{error term} $L_{+}$ on $\FF$ via
%\begin{align*}
%	L_{+}u(x)&= \sum_{y:\, \nabla_{x,y}u>0}b(x,y)B_{p}(\nabla_{x,y}u)%,\\
%	%L_{-}u(x)&= \sum_{y: \nabla_{x,y}u<0}b(x,y)B_{p}(\nabla_{x,y}u)
%	, \quad x\in X.
%\end{align*}
We state now the main result of this paper. %Recall that $B_p$ is skew-symmetric.

\begin{theorem}[Ground state representation]\label{thm:GSR}
	Let $p> 1$ and $0\leq u\in \FF(V)$ for some $V\sse X$. Then, we have
	\begin{align}\label{eq:GSRI}
		 h(u\phi)- \ip{Hu}{u\abs{\phi}^{p}}%\asymp h_{u}(\phi)
		 \asymp h_{u}(\phi), \qquad \phi\in C_c(V).%\\
		% \label{eq:GSRII}&\asymp h_{u,\s}(\phi), \qquad \phi\in C_c(X).
	\end{align}
	Furthermore, the equivalence becomes an equality if $p=2$.
\end{theorem}
%Note that the representation does not have any restriction on the potential part including the choice of the function $C$.

In many applications the function $u$ is assumed to be harmonic in $V\sse X$. In this case the representation in \eqref{eq:GSRI} reduces to
\[h(u\phi)\asymp h_{u}(\phi), \qquad \phi\in C_c(V).\] 

A further consequence of \eqref{eq:GSRI} is, that the corresponding left-hand side is non-negative, i.e,
\[h(u\phi)\geq  \ip{Hu}{u\abs{\phi}^{p}}, \qquad \phi\in C_c(V).\]
This inequality is known as \emph{Picone's inequality}, see \cite{AllegrettoHuangPicone, AM,  AmghibechPicone, BF14, F:AAP, FS08, ParkKimChung, Picone, PTT08} for applications of this inequality in various contexts.

%From the equivalence in Theorem~\ref{thm:GSR}, we get as a consequence the following estimates. Set $\FF_{\cap}(V)=\FF_{b, \abs{\cdot}^{p-1}\sgn(\cdot)}(V)\cap \FF_{b, B}(V)$ for $V\sse X$, and  denote by $\HH$ the classical $p$-Schrödinger operator.
%\begin{corollary}\label{cor:GSR0}
%	Let $p> 1$ and $V\sse X$. If $\alpha\abs{B(t)t}\leq\abs{t}^{p}$ and $\alpha\abs{C(t)t}\leq\abs{t}^{p}$, $t\in \RR,$ for some non-negative constant $\alpha$, then there is a positive constant $c_{p}$ such that for all $0\leq u\in \FF_{\cap}(V)$,
%	\begin{align}\label{eq:GSRI0}
%		\alpha\bigl(h(u\phi) - \ip{Hu}{u\abs{\phi}^{p}}\bigr)+ \ip{\alpha Hu-\HH u}{u\abs{\phi}^{p}}\leq c_p h_{u}(\phi),\quad \phi\in C_c(V).
%	\end{align}
%	
%	 If $\alpha\abs{B(t)t}\geq\abs{t}^{p}$ and $\alpha\abs{C(t)t}\geq\abs{t}^{p}$, $t\in \RR,$ the reversed inequality in \eqref{eq:GSRI0} holds true.
%\end{corollary}

From the inequalities in Theorem~\ref{thm:GSR}, we get as consequences estimates between the energy associated with the Schrödinger operator and other functionals, which are usually also referred to as \emph{simplified energies} (see e.g. \cite{DP16, PTT08}). They all are called simplified, because they consist of non-negative terms only, and the difference operator $\nabla$ applies either to $u$ or $\phi$ but not to the product $u\cdot \phi$.

We set on $C_c(V)$,
\begin{align*}
	h_{u,1}(\phi):= \sum_{x,y\in X}b(x,y) (u(x)u(y))^{p/2}\abs{\nabla_{x,y}\phi}^{p},%\\
\end{align*}
%where the summand in the latter sum on the right-hand side is understood to be zero if $(u(x)u(y))^{1/2}\abs{\nabla_{x,y}\phi}+ (\abs{\phi(x)}+ \abs{\phi(y)})\abs{\nabla_{x,y}u}/2=0$.
and for $p\geq 2$, we define on $C_c(V)$
\[h_{u,2}(\phi):=\sum_{x,y\in X}b(x,y)u(x)u(y)\abs{\nabla_{x,y}u}^{p-2} \left(\frac{\abs{\phi(x)}+ \abs{\phi(y)}}{2}\right)^{p-2}\abs{\nabla_{x,y}\phi}^{2}.\]
%%
%Then, the ground state representation equivalence \eqref{eq:GSRI}, yields in the form sense
%\[h(u\phi)- (muHu)(\phi)\asymp h_{u}(\phi), \qquad \phi\in C_c(X).\]
%Moreover, \eqref{eq:GSRI_p<2} then reads
%\[C_1(p)h_{u,1}(\phi)\geq h(u\phi)- (muHu)(\phi) \geq C_{2}(p)h_{u,2}(\phi), \qquad \phi\in C_c(X),\]
%and \eqref{eq:GSRI_p>2} changes to
%\[C_1(p)h_{u,1}(\phi)\leq h(u\phi)- (muHu)(\phi) \leq C_{2}(p)h_{u,2}(\phi), \qquad \phi\in C_c(X).\]
The following corollary is an immediate consequence of Theorem~\ref{thm:GSR}.%is formulated for the case $B_p=\abs{\cdot}^{p-1}\sgn$. Mutatis mutandis, it does also hold for more general choices of $B_p$ using \eqref{eq:GSRI0} instead of \eqref{eq:GSRI}.
\begin{corollary}\label{cor:GSR}
	Let $p>1$. If $1<p\leq 2$, then there is a positive constant $c_{p}$ such that for all $0\leq u\in \FF(V)$
	\begin{align}\label{eq:GSRI_p<2}
		h(u\phi) - \ip{Hu}{u\abs{\phi}^{p}}\leq c_p h_{u,1}(\phi),\qquad \phi\in C_c(V),
	\end{align}
	and if $p\geq 2$ the reversed inequality in \eqref{eq:GSRI_p<2} holds true, i.e., 
	\begin{align}\label{eq:GSRI_p>2}
	h(u\phi) - \ip{Hu}{u\abs{\phi}^{p}}\geq c_p h_{u,1}(\phi),\qquad \phi\in C_c(V).
		\end{align}
 Furthermore, both inequalities become equalities if $p=2$.
	
	Moreover, if $p\geq 2$, we have for all $0\leq u\in \FF(V)$,
	\begin{align}\label{eq:GSRI_p>2Triangle}
	h(u\phi)- \ip{Hu}{u\abs{\phi}^{p}} \asymp h_{u,1}(\phi)+h_{u,2}(\phi), \qquad \phi\in C_c(V).
		\end{align}
\end{corollary}
The statements in Theorem~\ref{thm:GSR} and Corollary~\ref{cor:GSR} will follow mainly by pointwise inequalities without summation. Then, we will sum over $X\times X$ and use Green's formula to obtain the results. The elementary inequalities are basically given in the upcoming lemma, Lemma~\ref{lem:preGSRep}.

The proof does not include the case $p=1$. This is because we use a quantification of the strict convexity of the mapping $x\mapsto \abs{x}^p$, $p> 1$.

\subsection{Some Remarks on the Main Result}\label{sec:discuss}

%\begin{remark}
%\Hmm{HIER weiter -> den Kommentar rausnehmen???}
%Let us have a closer look on the summands in the definition of the functionals $h-\ip{uHu}{\abs{\cdot}^{p}}, h_{u}$ and $h_{u,2}$ in the case $1<p<2$. 
%
%Firstly, consider \eqref{eq:GSRI}: Per definitionem, the summands of $h_u$ are understood to be zero if $\abs{\nabla_{x,y}u\phi}+\abs{\phi_{u}(x,y)}\abs{\nabla_{x,y}u}=0$. However, this implies that both $\abs{\nabla_{x,y} u\phi}$ and $\abs{\phi_{u}(x,y)}\abs{\nabla_{x,y}u}$ have to be zero. The first implies that the corresponding two summands for $h$ have to be zero, and the second implies that corresponding two summands in $\ip{uHu}{\abs{\phi}^{p}}$ are zero as well. 
%
%Secondly, consider \eqref{eq:GSRI_p<2}:  Here a similar argumentation shows that if we have $(u(x)u(y))^{1/2} \abs{\nabla_{x,y}\phi} + (\abs{\phi(x)}+ \abs{\phi(y)})\abs{\nabla_{x,y}u}=0$, then the corresponding summands in $h-\ip{uHu}{\abs{\cdot}^{p}}$ are zero as well.
%%
%%
%%That we interpret summand on the right-hand side in \eqref{eq:GSRI} to be zero if $1<p<2$ and $\abs{\nabla_{x,y}u\phi}+\abs{\phi_{u}(x,y)}\abs{\nabla_{x,y}u}=0$, has its legitimacy in the fact that the latter would imply that both $\abs{\nabla_{x,y} u\phi}$ and $\abs{\phi_{u}(x,y)}\abs{\nabla_{x,y}u}$ have to be zero. But this implies that the corresponding summand on the left-hand side is zero. 
%\end{remark}

\begin{remark}[Comparison with the local non-linear analogue]\label{rem1} We compare our ground state representation with results in \cite{PTT08}. Similar results associated with weighted $p$-Schrödinger operators can be found in \cite{PR15}.

Fix $p\in (1,\infty)$ and a domain $\Omega \sse \RR^{d}$. Let $u\in W^{1,p}_{\mathrm{loc}}(\Omega)$ and $\Delta (u):=- \Div (\abs{\nabla u}^{p-2} \nabla u )$ be the $p$-Laplacian on $\Omega$. Furthermore, let $V\in L^{\infty}_{\mathrm{loc}}(\Omega)$. The corresponding energy functional to the Schrödinger operator $\Delta + V$ is given by
\[Q(\phi):= \int_{\Omega}\abs{\nabla \phi}^{p}+ V\abs{\phi}^{p}\dd x, \qquad \phi \in C^{\infty}_{c}(\Omega).\]

Then, by \cite[Lemma~2.2]{PTT08}, we have the following: If $u$ is a positive harmonic function of $\Delta +V$ in the weak sense, i.e., $\int_{\Omega}\abs{\nabla u}^{p-2}\nabla u \cdot \nabla \phi + V \abs{u}^{p-2}u \phi \dd x = 0$ for all $\phi\in C_{c}^{\infty}(\Omega)$, then
\begin{align}\label{eq:PTT1}
Q(u\phi)\asymp \int_{\Omega} u^{2}\abs{\nabla \phi}^{2}\bigl(u \abs{\nabla \phi}+ \phi \abs{\nabla u}  \bigr)^{p-2}\dd x, \qquad 0\leq \phi\in C^{1}_{c}(\Omega).	
\end{align}
In particular, for $p>2$, we have
\begin{align}\label{eq:PTT2}
	Q(u\phi)\asymp \int_{\Omega} u^{p} \abs{\nabla \phi}^{p}+ u^{2}\abs{\nabla u}^{p-2}\phi^{p-2} \abs{\nabla \phi}^{2}  \dd x, \qquad 0\leq \phi\in C^{1}_{c}(\Omega).
\end{align}
In the case of $1<p<2$, we have by \cite[Remark~1.12]{PTT08} that
\begin{align}\label{eq:PTT3}
	\int_{\Omega} u^{2}\abs{\nabla \phi}^{2}\bigl(u \abs{\nabla \phi}+ \phi \abs{\nabla u}  \bigr)^{p-2}\dd x \leq \int_{\Omega} u^{p}\abs{\nabla \phi}^{p}.
\end{align}
Now, we do the comparison: In the continuum, domains of $\RR^{d}$ are considered. On graphs, we can take any subset of the graph.

Recall that $u$ is harmonic. It is very easy to compare $h_{u}(\phi)$ with the right-hand side in \eqref{eq:PTT1}, see Table~\ref{table1}. %Note that $h_{u,2}$ is only a one-sided estimate for $h$, see \eqref{eq:GSRI_p<2} and \eqref{eq:GSRI_p>2}, whereas we have an equivalence in \eqref{eq:PTT1}.

\begin{table}[h]
\caption{Comparison of the terms in the right-hand side (RHS) of  \eqref{eq:PTT1} with  $h_{u}(\phi)$.}
\label{table1}
 \begin{tabular}{c|c}
\rule{0pt}{2.5ex} 
RHS of \eqref{eq:PTT1} & $h_{u}(\phi)$ \\ 
\hline 
\rule{0pt}{2.5ex} 
$u^{2}\abs{\nabla \phi}^{2}$ & $u(x)u(y)\abs{\nabla_{x,y}\phi}^{2}$ \\ 
\hline 
\rule{0pt}{2.5ex} 
$u \abs{\nabla \phi}+\phi \abs{\nabla u}$ &  $(u(x)u(y))^{1/2}\abs{\nabla_{x,y}\phi}+\frac{1}{2}(\abs{\phi(x)}+ \abs{\phi(y)})\abs{\nabla_{x,y}u}$  \\ 
%\rule{0pt}{2.5ex} 
%$u \abs{\nabla \phi}$ &  $(u(x)u(y))^{1/2}\abs{\nabla_{x,y}\phi}$  \\ 
%\hline 
%\rule{0pt}{2.5ex} $\phi \abs{\nabla u}$  & $(\abs{\phi(x)}+ \abs{\phi(y)})\abs{\nabla_{x,y}u}$ \\ 
\end{tabular}
 \end{table} 
 This motivates to call the simplified energy $h_{u}$ the analogue to the simplified energy in the local non-linear case. Note that in the continuum, we only consider non-negative compactly supported functions $\phi$, whereas on graphs, we allow $\phi$ to take negative values. Thus, the version in the continuum contains hidden moduli of $\phi$.

%%Moreover, we see that \eqref{eq:PTT1} has almost the same structure as \eqref{eq:GSRI}, see Table~\ref{table0}.
%%
%%\begin{table}[h]
%%\caption{Comparision of the terms in the right-hand side (RHS) of  \eqref{eq:PTT1} with  $h_{u}(\phi)$.}
%%\label{table0}
%% \begin{tabular}{c|c}
%%\rule{0pt}{2.5ex} 
%%RHS of \eqref{eq:PTT1} & $h_{u}(\phi)$ \\ 
%%\hline 
%%\rule{0pt}{2.5ex} 
%%$u^{2}\abs{\nabla \phi}^{2}$ & $u(x)u(y)\abs{\nabla_{x,y}\phi}^{2}$ \\ 
%%\hline 
%%\rule{0pt}{2.5ex} 
%%$u \abs{\nabla \phi}+\phi \abs{\nabla u}$ &  $\abs{\nabla_{x,y}u\phi}+\abs{\phi_{u}(x,y)}\abs{\nabla_{x,y}u}$  \\ 
%%%\rule{0pt}{2.5ex} 
%%%$u \abs{\nabla \phi}$ &  $(u(x)u(y))^{1/2}\abs{\nabla_{x,y}\phi}$  \\ 
%%%\hline 
%%%\rule{0pt}{2.5ex} $\phi \abs{\nabla u}$  & $(\abs{\phi(x)}+ \abs{\phi(y)})\abs{\nabla_{x,y}u}$ \\ 
%%\end{tabular}
%% \end{table} 

Furthermore, we see that the equivalence \eqref{eq:PTT2} has the same structure as the equivalence \eqref{eq:GSRI_p>2Triangle}. For a comparison of  $h_{u,1}(\phi)+ h_{u,2}(\phi)$ with the right-hand side in \eqref{eq:PTT2} see Table~\ref{table2}. 

\begin{table}[h]
\caption{Comparison of the terms in the right-hand side (RHS) of  \eqref{eq:PTT2} with  $h_{u,1}(\phi)+h_{u,2}(\phi)$.}
\label{table2}
 \begin{tabular}{c|c}
\rule{0pt}{2.5ex} RHS of \eqref{eq:PTT2} & $h_{u,1}(\phi)+ h_{u,2}(\phi)$ \\ 
\hline 
\rule{0pt}{2.5ex} $u^{p}\abs{\nabla \phi}^{p}$ & $(u(x)u(y))^{p/2}\abs{\nabla_{x,y}\phi}^{p}$ \\ 
\hline 
\rule{0pt}{3ex} $u^{2}\abs{\nabla u}^{p-2}\phi^{p-2} \abs{\nabla \phi}^{2}$ &   $u(x)u(y)\abs{\nabla_{x,y}u}^{p-2} \bigl(\frac{1}{2}(\abs{\phi(x)}+ \abs{\phi(y)})\bigr)^{p-2}\abs{\nabla_{x,y}\phi}^{2}$  \\  
\end{tabular}
 \end{table} 
 Furthermore, we see that the estimate in \eqref{eq:PTT3} together with \eqref{eq:PTT1} has the same structure as the upper bound \eqref{eq:GSRI_p<2}.

It should be mentioned that the strategy to prove the ground state representation in \cite{PTT08} and here are similar. There, an elementary equivalence is the key ingredient and then a Picone identity is used. Here, we use different elementary equivalences and the Green's formula. However, the proof of the elementary equivalences in the discrete is technically much harder as the proof of the corresponding one in the continuum. Thus, the differences above might come from the fact that in the continuum we have a Picone identity (see \cite[Section~2]{PTT08}) which is established via the chain rule. Whereas in the discrete, we only have a one-sided Picone inequality and the missing of a chain rule in general. A general version of this one-sided Picone inequality is discussed in a follow-up paper by the author \cite{F:AAP}, see also \cite{BF14}.

Moreover, in \cite[Proposition~5.1]{PTT08} it was shown that for $p>2$ both summands in the integral in \eqref{eq:PTT2} are needed in general for an upper bound. We expect that the same holds true on graphs, i.e., we expect that both $h_{u,1}$ and  $h_{u,2}$ are needed in general for an upper bound of $h$.  
\end{remark}

\begin{remark}[Discussion of the constants]
By comparing Theorem~\ref{thm:GSR} with \cite[Proposition~2.3]{FS08} and Lemma~\ref{lem:preGSRep} (the lemma below) with \cite[Lemma~2.6]{FS08}, we see that $c_p$ in \eqref{eq:GSRI_p>2} can be stated explicitly as a minimiser, i.e., for $p\geq 2$
\[c_p= \frac{1}{2}\min_{t\in (0,1/2)}\bigl( (1-t)^{p}-t^{p}+pt^{p-1}\bigr)\in (0,1/2].\]
Note that $c_2=1/2$. Moreover, we expect that the best constants in Theorem~\ref{thm:GSR} are between $0$ and $1$.
\end{remark}

\begin{remark}[Hardy inequality]
	In \cite{DP16} the non-linear ground state representation of \cite{PTT08} was used to prove optimality of certain $p$-Hardy weights associated with $p$-Schrödinger operators on domains in $\RR^d$. The discrete counterpart does hold as well using the here presented discrete ground state representation and are topic of a follow-up paper by the author, \cite{F:Opti}. This generalises the results of the linear case in \cite{KePiPo2} to $p\neq 2$. A consequence of the ground state representation and some results in this yet unpublished paper is that the improved $p$-Hardy inequality on $\NN$ in \cite{FKP} is indeed optimal.
\end{remark}

\begin{example}[Standard $p$-Laplacian on $\NN_{0}$]
	Here, we calculate the representation for one of the simplest cases: for the graph $b$ on $\NN$ with $b(n,m)=1$ if $\abs{n-m}=1$ and $b(n,m)=0$ elsewhere for all $n,m\in\NN$. 
	
	The standard (or combinatorial) $p$-Laplacian $\Delta$ for real valued functions on $ \NN_{0}=\NN\cup \set{0} $ is given by
\[\Delta f(n)=\sum_{m=n\pm1}\sgn\left( f(n)-f(m) \right)\abs{f(n)-f(m)}^{p-1}\]
for all functions $f\in C(\NN)$ and $ n\ge 1 $. The corresponding energy functional reads then as
\[h(\phi)= \frac{1}{2}\sum_{n\sim m}^\infty \abs{\phi(n)-\phi(m)}^p =\sum_{n=1}^\infty \abs{\phi(n)-\phi(n-1)}^p,\]
for all $\phi\in C_{c}(\NN)$. From \cite[Proposition~4]{FKP} it follows that $u\in\FF$ defined via $u(n)=n^{(p-1)/p}$ is a positive superharmonic function such that $\Delta u= w u^{p-1}$, where $w$ is the improved $p$-Hardy weight in \cite[Theorem~1]{FKP}. Let $q:=p/(p-1)$ and $\alpha(n):=(1-1/n)^{1/q}$, $n\in \NN$. Then, the equivalence \eqref{eq:GSRI} reads as follows: for all $\phi\in C_c(\NN)$, we have
\begin{align*}
	&\sum_{n=1}^\infty \abs{\phi(n)-\phi(n-1)}^p- w(n)\abs{\phi(n)}^{p}\\
	&\asymp \sum_{n=2}^{\infty}\frac{1}{\alpha^{p-1}(n)}\bigl(\alpha(n)\phi(n)-\phi(n-1)\bigr)^{2}\\
	&\quad \cdot\left( \alpha^{1/2}(n)\abs{\alpha(n)\phi(n)-\phi(n-1)}+ \frac{\alpha(n)\abs{\phi(n)}+\abs{\phi(n-1)}}{2}\bigl(1-\alpha(n)\bigr) \right)^{p-2}.
\end{align*}
If $p=2$, then the equivalence is an equality and gives exactly the result of \cite[Theorem~1]{KS21}. 

Moreover, the inequality \eqref{eq:GSRI_p>2} ($p\geq 2$) in Corollary~\ref{cor:GSR} is here
\begin{align*}
	\sum_{n=1}^\infty \abs{\phi(n)-\phi(n-1)}^p- w(n)\abs{\phi(n)}^{p}
	\geq c_p \sum_{n=2}^{\infty}\frac{1}{\alpha^{p/2}(n)}\abs{\alpha(n)\phi(n)-\phi(n-1)}^{p}
\end{align*}
for all $\phi\in C_c(\NN)$. By \eqref{eq:GSRI_p<2}, the reversed inequality holds for $1<p\leq 2$. 
\end{example}

\subsection{Elementary Inequalities and Equivalences}\label{sec:elementary}
We need the following quantification of the strict convexity of the mapping $x\mapsto \abs{x}^p$, $p> 1$. In the following lemma, $\langle \cdot, \cdot \rangle_{\RR^{n}}$ denotes the standard inner product in $\RR^{n}$.
\begin{lemma}[Lindqvist's lemma, Lemma~4.2 in \cite{Lindqvist90}]\label{lem:Lindqvist90}
	Let $a, b\in \RR^n$. Then, for all $p\geq 2$ we have
	\[\abs{a}^p- \abs{b}^p\geq p\abs{b}^{p-2}\langle b, a-b \rangle_{\RR^{n}} + c_p\abs{a-b}^p ,\]
where $c_p=1/(2^{p-1}-1)> 0$.	If $1<p<2$, then
		\[\abs{a}^p- \abs{b}^p\geq p\abs{b}^{p-2}\langle b, a-b \rangle_{\RR^{n}} + c_p\frac{\abs{a-b}^2 }{(\abs{a}+\abs{b})^{2-p}},\]
		where $c_p=3p(p-1)/16>0$, and the fraction is interpreted to be zero if $a=b=0$.	
\end{lemma}
In the previous lemma, the constant $c_p$ does not seem to be optimal. However, this is not important for our further investigations.

The next lemma is the most important tool in order to derive the ground state representations, Theorem~\ref{thm:GSR} and Corollary~\ref{cor:GSR}. %{\color{red}In the statement, we use the notation $f\asymp g$ for some functions $f$ and $g $ defined on $I\sse \RR^{d}$, $d\in\NN$, which means that there exists a constant $C$ such that $C^{-1}f\leq g\leq C f$ on $I$. The functions $f$ and $g$ are then called \emph{equivalent}.}

\begin{lemma}[Fundamental inequalities and equivalences]\label{lem:preGSRep}
	Let $a\in \RR$, $0\leq t\leq 1$, and $p> 1$. Then we have 
	\begin{align}\label{eq:preGSRep2}
		 \abs{a-t}^p -(1-t)^{p-1}(\abs{a}^p-t)\asymp t\abs{a-1}^2(\abs{a-t}+1-t)^{p-2},
	\end{align}
	where the right-hand side is understood to be zero if $1<p<2$ and $a=t=1$.
	
	Moreover, we have %for all $C\geq 2$,
	\begin{align}\label{eq:preGSRep1}
		\abs{a-t}+1-t\asymp  t^{1/2}\abs{a-1}+ (1-t)\frac{\abs{a}+1}{2},
	\end{align}
	where the right-hand side is an upper bound with optimal constant $c=2$, and it is a lower bound with optimal constant $c=1/2$. 
%		and for all $0\leq C\leq 1/2$,
%	\begin{align}\label{eq:preGSRep1Neu}
%		\abs{a-t}+1-t\geq C\left( t^{1/2}\abs{a-1}+ (1-t)\frac{\abs{a}+1}{2}\right),
%	\end{align}
%	and neither \eqref{eq:preGSRep1} nor \eqref{eq:preGSRep1Neu} are true for all $(t,a)\in [0,1]\times \RR$ if $1/2<C<2$. 
	
%	In particular, we have
%		\begin{multline}\label{eq:preGSRep2Neu}
%		 \abs{a-t}^p -(1-t)^{p-1}(\abs{a}^p-t)\\
%		 \asymp t\abs{a-1}^2\left(t^{1/2}\abs{a-1}+ (1-t)\frac{\abs{a}+1}{2}\right)^{p-2},
%	\end{multline}
%	where the right-hand side is understood to be zero if $1<p<2$ and $a=t=1$.

	Furthermore, if $1<p\leq 2$, then 
	\begin{align}\label{eq:preGSRep3}
		t\abs{a-1}^2\leq t^{p/2}\abs{a-1}^p(\abs{a-t}+1-t)^{2-p},
	\end{align}
	and for $p\geq 2$, the reserved inequality holds, i.e., 
	\begin{align}\label{eq:preGSRep4}
		t\abs{a-1}^2(\abs{a-t}+1-t)^{p-2}\geq t^{p/2}\abs{a-1}^p.
	\end{align}
	
	Moreover, we have the following refinement of the elementary inequality \eqref{eq:pTriangle}: for all $p\geq 0$, we have
	\begin{align}\label{eq:preGSRep5}
	\alpha^{p}+\beta^{p}\asymp	(\alpha + \beta )^{p}  ,\qquad \alpha, \beta \geq 0,
	\end{align}	
	where the right-hand side is an upper bound with optimal constant $c_p=2^{1-p}$ if $0\leq p\leq 1$ and $c_p=1$ if $p\geq 1$, and it is a lower bound with optimal constant $c_p=1$ for $0\leq p\leq 1$ and $c_p=2^{1-p}$ for $p\geq 1$. 
%
%	In particular, we have for all $p\geq 2$
%	\begin{multline}\label{eq:preGSRep5Neu}
%		 \abs{a-t}^p -(1-t)^{p-1}(\abs{a}^p-t)\\
%		 \asymp t^{p/2}\abs{a-1}^p + t\abs{a-1}^2(1-t)^{p-2}\left(\frac{\abs{a}+1}{2}\right)^{p-2}.
%	\end{multline}
\end{lemma}
We do not claim that the constants we get in \eqref{eq:preGSRep2} are optimal. We expect that they can be improved and that the best constants  should be either on the boundary of $[0,1]\times \RR$, or at $(t,0), (t,t), (t,1)$, $t\in [0,1]$. Moreover, we expect that the optimal constants are between $0$ and $2$. 
%Before we start with the technical proof, note that  . After a close look at a plot of the equivalence, we expect that the best constants should be either on the boundary of $[0,1]\times \RR$, or at $a\in \set{0,t,1}$, and are between $0$ and $2$. 

Also note that the inequalities \eqref{eq:preGSRep2} and \eqref{eq:preGSRep4} show that we improved an elementary one-sided result in \cite{FS08} for $p>2$.

Moreover, in the case of $1<p<2$, the "$\geq$"-inequality in \eqref{eq:preGSRep2} was proven in \cite[Lemma~3.3]{AM}. However, the basic strategy to prove the remaining inequalities in \eqref{eq:preGSRep2} up to a certain point will be the similar, i.e., we start the proof with the same substitution and then use the same Taylor-Maclaurin formula (confer this also with the proof of \cite[Lemma~4.2]{Lindqvist90}). 

Furthermore, note that \eqref{eq:preGSRep2} is false for $p=1$ as the left-hand side vanishes for $a>1\geq t>0$ but the right-hand side does not. A similar argument can also be made for \eqref{eq:preGSRep3}.

\begin{proof}[Proof of Lemma~\ref{lem:preGSRep}]	%Let $a\in \RR$. The case $a\in \set{\pm \infty}$ is treated at the very end of the proof.

\emph{Ad~\eqref{eq:preGSRep2}:}  Recall that we have to show that for $p>1$,
	\begin{align*}
		 \abs{a-t}^p -(1-t)^{p-1}(\abs{a}^p-t)\asymp t\abs{a-1}^2(\abs{a-t}+1-t)^{p-2}, \quad a\in \RR, 0\leq t\leq 1.
	\end{align*}
The strategy of the proof is as follows: We start with some simple special cases for which the equivalence can be shown very easily. Thereafter, we do a substitution to bring the equivalence in a simpler form for the remaining cases. Then, we divide $\RR$ into the three intervals $[1, +\infty)$, $(t,1)$, and  $(-\infty, t]$ for some $t\in [0,1]$. In the two intervals $[1, +\infty)$ and  $(-\infty, t]$, we then distinguish between proving lower bounds and upper bounds, as well as having $p> 2$ or $1<p<2$. In the remaining interval $(t,1)$, we show that we can deduce the equivalence from the validity of the equivalence in $[1, +\infty)$.
	
	\textbf{1. The three cases $t\in \set{0, 1}$, $a=t$, and $p=2$:} If $p=2$, then it is obvious that we have equality for all $a\in \RR$ and $t\in [0,1]$.
	
	An easy computation shows that we have indeed equality for $t\in \set{0,1}$. 
	
	If $a=t$, we have to show that for all $p> 1$
	\[-(1-t)^{p-1}(t^p-t)\asymp t(1-t)^p. \]
	Thus, let us consider the function
	\[q(t):= \frac{(1-t)^{p-1}(t-t^p)}{t(1-t)^p}= \frac{1-t^{p-1}}{1-t}.\]
	If $1<p<2$, then $t^{p-1} \geq t$ for $t\in (0,1)$, and thus, $q$ is decreasing. If $p > 2$, we have  $t^{p-1} \leq t$ for $t\in (0,1)$, and $q$ is increasing. Moreover, by L'H\^{o}pital's rule $q(1)=p-1$.
%	On the interval $(0,1)$, we can use the mean value theorem applied to $f(t)=t^{p-1}$, $0<t<1$, and get
%	\[(p-1) \inf_{t_0\in (t,1)} t_0^{p-2} \leq \frac{f(1)-f(t)}{1-t}= q(t) \leq (p-1)\sup_{t_0\in (t,1)} t_0^{p-2}.\]
	Hence, for $p> 2$ we have $1=q(0)\leq  q(t)\leq  q(1)=p-1$ and for $1<p<2$, we have $p-1=q(1)\leq q(t) \leq q(0)=1$.

	\textbf{2. The remaining cases $t\in (0, 1)$, $a\neq t$, and $p\neq 2$:} We do the following substitution: Set $\alpha :=(a-t)/(1-t)$, then we have to show that 
	\begin{align}\label{eq:preGSRepalpha}
		\abs{\alpha}^p- \frac{\abs{\alpha(1-t)+t}^p-t}{1-t}\asymp \frac{t(\alpha -1 )^2}{(\abs{\alpha}+1)^{2-p}}.
	\end{align}
	We will do this, by considering the following three cases separately
	\begin{itemize}
	\item $\alpha \geq 1$,
	\item $1>\alpha > 0$, and
	\item $ \alpha < 0$.
	\end{itemize}
	Furthermore, let 
	\[f_{\alpha}(t):= \frac{\abs{\alpha(1-t)+t}^p-t}{1-t}= \frac{\abs{\alpha + t(1-\alpha)}^p-t}{1-t}.\]
	Note that $f_{\alpha}(0)=\abs{\alpha}^{p}$.
	
	\textbf{2.1. The case $\alpha \geq 1$:} 
	The basic strategy is to use the Taylor-Maclaurin formula. Thus, let us calculate the first and the second derivatives with respect to $t$. Note that for $\alpha \geq 1$, we have $\abs{\alpha + t(1-\alpha)}=\alpha + t(1-\alpha)$. Hence, we calculate
	\[f'_{\alpha}(t)=\frac{p(1-\alpha)(\alpha + t(1-\alpha))^{p-1}-1}{1-t} + \frac{f_{\alpha}(t)}{1-t}, \]
	and using $\alpha+t(1-\alpha)= \beta +1$, where $\beta:= (\alpha -1)(1-t)\geq 0$, we get
	\begin{align}
		\nonumber f_{\alpha}''(t)&=\frac{p(p-1)(1-\alpha)^2(\alpha+t(1-\alpha))^{p-2}}{1-t}+ \frac{p(1-\alpha)(\alpha +t(1-\alpha))^{p-1}-1}{(1-t)^2} \\
		\nonumber &\quad + \frac{f'_{\alpha}(t)}{1-t}+\frac{f_{\alpha}(t)}{(1-t)^2}\\
		\nonumber &= \frac{(\alpha+t(1-\alpha))^{p-2}}{(1-t)^3}\Bigl( p(p-1)(\alpha-1)^2(1-t)^2 \Bigr. \\
		\nonumber &\quad \Bigl. - 2p(\alpha-1)(1-t)(\alpha+t(1-\alpha))+2(\alpha+t(1-\alpha))^2  \Bigr) - 2\frac{1-t+t}{(1-t)^3} \\
		\nonumber &=\frac{(\beta +1)^{p-2}}{(1-t)^3}\Bigl( p(p-1)\beta^2 - 2p\beta(\beta +1)+2(\beta +1)^2  \Bigr) - \frac{2}{(1-t)^3}\\
		\label{eq:f''alpha}&=\frac{(\beta +1)^{p-2}}{(1-t)^3}\Bigl( -(p-1)(2-p)\beta^2+ 2(2-p)\beta+2 \Bigr) - \frac{2}{(1-t)^3}\\
		\nonumber &=\frac{g(\beta)-2}{(1-t)^3},
	\end{align}
	where \[g(\beta):=\bigl((p-1)(p-2)\beta^2+ 2(2-p)\beta+2\bigr)(\beta +1)^{p-2}, \qquad \beta \geq 0.\]
%	Since $g(\beta)=((p-1)(p-2)\beta^2+ 2(2-p)\beta+2)$ is convex for $p>2$ and concave for $1<p<2$, and we have an extrema at $\beta =1/(p-1)$, we conclude for $p>2$, that $g(\beta)\geq g(1/(p-1))=p/(p-1)$, and for $1<p<2$, we have $g(\beta)\leq p/(p-1)$. 
 Let us analyse $g(\beta)$ for $\beta \geq 0$. Then, $g'(\beta)=p(p-1)(p-2)(\beta+1)^{p-3}\beta^2$, which is positive for $p>2$ and negative for $1<p<2$. Hence, $g(0)=2$ is a minimum for $p>2$ and a maximum for $1<p<2$. This implies that for all $t\in (0,1)$
 
 \begin{align*}
 	f''_{\alpha}(t)\begin{cases}\leq 0 \qquad \text{if } 1<p<2, \\
 	\geq 0 \qquad \text{if } p>2.  \end{cases}
 \end{align*}
 
%% hatte mich bei der Ableitung verrechnet ... 
% If $p>2$, then the situation is more difficult: Note that by substituting $\alpha -1 $ with $\beta$ in the above calculation, we see that
% \begin{align*}
% 	f_{\alpha}''(0)=\alpha^{p-2}\bigl( -(p-1)(2-p)(\alpha-1)^2+ 2(2-p)(\alpha -1)+2 \bigr) - 2 \geq 2-2=0.
% \end{align*}
% Moreover, 
% \[f_{\alpha}''(1)=  -\infty \leq 0.\]
% Thus, $f_{\alpha}''$ may change sign if $p>2$.
 		
 Now, we apply the Taylor-Maclaurin formula
 \[f_{\alpha}(t)= f_{\alpha}(0)+t f_{\alpha}'(0)+ \int_{0}^{t}(t-s)f_{\alpha}''(s)\dd s.\]
 Since $f_{\alpha}(0)=\alpha^{p}$, we have
  \begin{align}\label{eq:TaylorMaclaurin}
  \begin{aligned}
  	\abs{\alpha}^p&- \frac{\abs{\alpha(1-t)+t}^p-t}{1-t} = f_{\alpha}(0)- f_{\alpha}(t) \\
  	&= -t f_{\alpha}'(0)- \int_{0}^{t}(t-s)f_{\alpha}''(s)\,\dd s \\
  	&= t \bigl( (p-1)\alpha^{p} -p\alpha^{p-1}+1 \bigr)- \int_{0}^{t}(t-s)f_{\alpha}''(s)\,\dd s .
  \end{aligned}
  \end{align}
  This term will be analysed in the following for upper and lower bounds and different values of $p$.
  
  \textbf{2.1.1. Lower bound for $1<p<2$ and $\alpha \geq 1$:} Then $f''_{\alpha} \leq 0$ on $(0,1)$. Thus we conclude from \eqref{eq:TaylorMaclaurin},
  \[\abs{\alpha}^p- \frac{\abs{\alpha(1-t)+t}^p-t}{1-t}  \geq t \bigl( (p-1)\alpha^{p} -p\alpha^{p-1}+1 \bigr).\]
  Using Lindqvist's lemma, Lemma~\ref{lem:Lindqvist90}, with $b=\alpha$ and $a=1$, we see
  \[t \bigl( (p-1)\alpha^{p} -p\alpha^{p-1}+1 \bigr)= t \bigl( 1 -\alpha^{p} -p\alpha^{p-2}\alpha(1-\alpha) \bigr) \geq C_p  \frac{t(\alpha -1 )^2}{(\alpha+1)^{2-p}}.\]
  This is the desired lower bound in \eqref{eq:preGSRepalpha} for $1<p<2$ and $\alpha \geq 1$.
  
    \textbf{2.1.2. Upper bound for $p>2$ and $\alpha \geq 1$:} Then $f''_{\alpha} \geq 0$ on $(0,1)$. Thus we conclude from \eqref{eq:TaylorMaclaurin},
 \[\abs{\alpha}^p- \frac{\abs{\alpha(1-t)+t}^p-t}{1-t} \leq t \bigl( (p-1)\alpha^{p} -p\alpha^{p-1}+1 \bigr).\]
  Hence, it remains to show that there exists $C_p >0$ such that
  \[ \bigl( (p-1)\alpha^{p} -p\alpha^{p-1}+1 \bigr) \leq C_p  (\alpha -1 )^2(\alpha+1)^{p-2}.\]	
For any positive constant $C_p$ we have  using $(1+\alpha^{-1})^{p-2}\geq 1$,
    \begin{align*}
  	j(\alpha)&:= \alpha^{p-2}\Bigl(\bigl( (p-1)\alpha^{2} -p\alpha+\alpha^{2-p}\bigr ) - C_p(\alpha -1)^2(1+\alpha^{-1})^{p-2}\Bigr) \\
  	&\leq \alpha^{p-2}\Bigl(\bigl( (p-1)\alpha^{2} -p\alpha +\alpha^{2-p}\bigr)  - C_p(\alpha -1)^2\Bigr)\\
  	&= \alpha^{p-2}\Bigl(\bigl(p-1-C_p  \bigr)\alpha^{2} +\bigl(2C_p-p \bigr)\alpha+\alpha^{2-p}-C_p  \Bigr).
  \end{align*}
  Let $g(\alpha):= \bigl(p-1-C_p  \bigr)\alpha^{2} +\bigl(2C_p-p \bigr)\alpha+\alpha^{2-p}-C_p $ for $\alpha >0$,
  then 
     \begin{align*}
   		g'(\alpha)= 2(p-1-C_p)\alpha +(2C_p-p )- (p-2)\alpha^{1-p}
   \end{align*}
 has a root at $\alpha= 1$. If we can show that $g$ is concave on $[1,\infty]$, then $g(1)=0$ is a maximum. 
 	Since
 \begin{align*}
	  	g''(\alpha)= 2(p-1-C_p)+ (p-2)(p-1)\alpha^{-p}\leq 2(p-1-C_p)+ (p-2)(p-1),
	\end{align*}  
   $g$ is concave on $[1,\infty]$ for all $C_p\geq p(p-1)/2$, we found a possible constant such that $j(\alpha)\leq 0$. In other words, we have the desired upper bound for $p>2$. However, it is obvious that the constant can be improved.
  
  \textbf{2.1.3. Upper bound for $1<p<2$ and $\alpha\geq 1$:} For $1\leq p\leq 2$, the function $\abs{\cdot}^{p-1}$ is concave on $(0,\infty)$, thus 
  \[\abs{\alpha(1-t)+t}^{p-1}\geq (1-t)\alpha^{p-1}+t.\]
  Using this estimate in the left-hand side of \eqref{eq:preGSRepalpha}, we get
  \begin{align}\label{eq:preGSRepConcave}
  \abs{\alpha}^p- \frac{\abs{\alpha(1-t)+t}^p-t}{1-t} \leq t (\alpha^{p-1}-1)(\alpha-1)
  \end{align}
  Define for $\alpha \geq 1$, 
	\[g(\alpha):= (\alpha + 1)^{2-p}(\alpha^{p-1}-1)- \alpha +1,\] 
	then 
	\begin{align*}
		g(\alpha) = (\alpha^{p-1}-1)\alpha^{2-p}(1+\alpha^{-1})^{2-p}- \alpha +1 
		=(\alpha - \alpha^{2-p})\sum_{k=0}^{\infty}\binom{2-p}{k}\alpha^{-k} -\alpha +1.
	\end{align*}
	Since for all $k\in 2\NN $, $1\leq p \leq 2$ and $\alpha \geq 1$, we have
	\[\binom{2-p}{k}\alpha^{-k}+ \binom{2-p}{k+1}\alpha^{-k-1}\leq 0,\]
	we get for all  $1\leq p \leq 2$ and $\alpha \geq 1$,
	\begin{align*}
		g(\alpha)\leq (\alpha - \alpha^{2-p})\bigl(1+\frac{2-p}{\alpha}\bigr)  -\alpha +1 
		= (2-p)+1- \alpha^{2-p}- (2-p)\alpha^{1-p}  
		=:l(\alpha).
	\end{align*}
Since $l'(\alpha)=(2-p)((p-1)-\alpha)\alpha^{-p}\leq 0$ for $\alpha \geq 1$ and $1\leq p\leq 2$, we get
\[g(\alpha)\leq l(\alpha)\leq l(1)=0.\]	
Thus, using that $g\leq 0$ on $[1,\infty]$ results in \eqref{eq:preGSRepConcave} in
\[(\alpha^{p-1}-1)(\alpha-1) \leq \frac{(\alpha -1 )^2}{(\alpha+1)^{2-p}}.\]
This results in the right-hand side of \eqref{eq:preGSRepalpha} with constant $1$.

\textbf{2.1.4. Lower bound for $p>2$ and $\alpha \geq 1$:} For $p\geq 2$, the function $\abs{\cdot}^{p-1}$ is convex on $(0,\infty)$, thus 
  \[\abs{\alpha(1-t)+t}^{p-1}\leq (1-t)\alpha^{p-1}+t.\]
  Using this estimate in the left-hand side of \eqref{eq:preGSRepalpha}, we get
  \begin{align}\label{eq:preGSRepConvex}
  \abs{\alpha}^p- \frac{\abs{\alpha(1-t)+t}^p-t}{1-t} \geq t (\alpha^{p-1}-1)(\alpha-1)
  \end{align}
  Define for $\alpha \geq 1$, and some constant $C_p>0$,
	\[g(\alpha):= \alpha^{p-1}-1- C_p(\alpha -1)(\alpha + 1)^{p-2},\] 
	then
	\[g'(\alpha)=\alpha^{p-2}\bigl(p-1 -C_p\bigl( \bigl(1+ \frac{1}{\alpha}\bigr)^{p-2} + (p-2)\bigl(1+ \frac{1}{\alpha}\bigr)^{p-3}\bigl(1- \frac{1}{\alpha}\bigr) \bigr)   \bigr).\]
	If $p\geq 3$, then
	\[\ldots \geq \alpha^{p-2}\bigl(p-1 -C_p( 2^{p-2} + (p-2)2^{p-3} )   \bigr). \]
Choosing $C_p= 2^{3-p}(p-1)/p$, we get $g'\geq 0$ on $[1,\infty)$. In particular,
\[g(\alpha)\geq g(1)= 0.\]
Thus, for $p\geq 3$, 
\begin{align}\label{eq:Convexp>3}
		(\alpha^{p-1}-1)(\alpha-1) \geq C_p (\alpha -1)^2(\alpha +1)^{p-2}.
\end{align}	
 If $2\leq p\leq 3$, then 
	\[g'(\alpha) \geq \alpha^{p-2}\bigl(p-1 -C_p(2^{p-2}+p-2)  \bigr). \] 
 Choosing $C_p=(p-1)/(2^{p-2}+p-2)$, we get $g'\geq 0$ on $[1,\infty)$.
 Thus, for $2\leq p\leq 3$, 
\begin{align}\label{eq:Convexp<3}
		(\alpha^{p-1}-1)(\alpha-1) \geq C_p (\alpha -1)^2(\alpha +1)^{p-2}.
\end{align}
Applying \eqref{eq:Convexp>3} and \eqref{eq:Convexp<3} to \eqref{eq:preGSRepConvex}, results in the right-hand side of \eqref{eq:preGSRepalpha}.

  Moreover, this was the last puzzle stone to show \eqref{eq:preGSRepalpha} for $\alpha \geq 1$ and all $1<p<\infty$.
  
	\textbf{2.2. The case $0 <\alpha < 1$:} We have shown that \eqref{eq:preGSRepalpha} holds for all $\alpha > 1$ and $t\in (0,1)$. Then it holds in particular for $s=1-t$, i.e.,
	\begin{align*}
				\abs{\alpha}^p- \frac{\abs{\alpha s+1-s}^p-(1-s)}{s}\asymp \frac{(1-s)(\alpha -1 )^2}{(\abs{\alpha}+1)^{2-p}}.
	\end{align*}
	Now, for any $\alpha >1$ let $\beta:=1/\alpha \in (0,1)$. Then, we get by multiplying both sides with $\beta^{p}s/(1-s)$,
	\begin{align*}
				\abs{\beta}^p- \frac{\abs{\beta (1-s)+s}^p-s}{1-s}\asymp \frac{s(\beta -1 )^2}{(\abs{\beta}+1)^{2-p}},
	\end{align*}
	which is the desired equivalence. 
	
	\textbf{2.3. The case $\alpha < 0$:} Set $\beta:= -\alpha$. Then, substituting into \eqref{eq:preGSRepalpha}, we have to show that for all $\beta > 0$ and $t\in (0,1)$,
	\begin{align}\label{eq:preGSRepbeta}
			\abs{\beta}^p- \frac{\abs{\beta(1-t)-t}^p-t}{1-t}\asymp \frac{t(\beta +1 )^2}{(\abs{\beta}+1)^{2-p}}= t(\beta +1)^{p}.
	\end{align}
	
	We have
	\[\abs{\beta}^p- \frac{\abs{\beta(1-t)-t}^p-t}{1-t} =\abs{\beta}^p- \frac{\abs{\beta(1-t)+t}^p-t}{1-t}+ g_{t}(\beta),\]
	where
	\[g_{t}(\beta):= \frac{1}{1-t}\bigl((\beta(1-t)+t)^{p} - \abs{\beta (1-t)-t}^{p}  \bigr), \qquad \beta >0,\,\, t\in (0,1).\]

	Before we continue with the estimates, let us note that 
	\[g_t\geq 0 \quad\text{and}\quad g_t'\geq 0.\] The first inequality can be seen as follows: let $\gamma > 0$. Firstly assume that $\gamma > t$. Then,
	\begin{align*}
		(\gamma +t)^{p} - (\gamma -t)^{p}= 2\gamma^{p}\sum_{k\in 2\NN-1}\binom{p}{k}\Bigl(\frac{t}{\gamma }\Bigr)^{k}> 0. 
	\end{align*}	
Secondly, if $\gamma \leq t$, then a similar calculation can be done to get the desired inequality (factor $t$ out of the sum and use  the binomial theorem).
	
Note that for all $p\geq 1$,
	 \[g_{t}'(\beta)=p\bigl( \abs{\beta (1-t)+t}^{p-1}- \abs{\beta (1-t)-t}^{p-1}\sgn (\beta (1-t)-t)    \bigr)\geq 0.\]	

	Now we continue with showing \eqref{eq:preGSRepbeta}: By the first parts of the proof, i.e., the proof of \eqref{eq:preGSRepalpha}, we have that for all $\beta > 0$, 
	\begin{align}\label{eq:preGSRbeta1}
		\abs{\beta}^p- \frac{\abs{\beta(1-t)+t}^p-t}{1-t} \asymp \frac{t(\beta -1 )^2}{(\abs{\beta}+1)^{2-p}}.
	\end{align}

	The strategy for the upper bound will be as follows: Clearly, $(\beta -1)^2\leq (\beta +1)^2$ for all $\beta > 0$. If we apply this estimate to \eqref{eq:preGSRbeta1}, we are left to show that also
	\[g_t(\beta)\leq C_p t(\beta+1)^p,\]
	for some positive constant $C_p$ in order to show the upper bound in \eqref{eq:preGSRepbeta}. 
	
	Let us turn to the strategy for the lower bound: It is obvious, that there does not exists a positive constant $C_p$ such that $(\beta -1)^2\geq C_p (\beta +1)^2$ since the left-hand side has a root at $\beta=1$. However, fix $0< \epsilon < 1$, then we clearly have for all $\beta\in (0,\infty)\setminus (1-\epsilon, 1+\epsilon)$ that $(\beta -1)^2\geq C_{p, \epsilon} (\beta +1)^2$ for some constant $C_{p,\epsilon}>0$. Since $g\geq 0$, we have the desired lower bound of \eqref{eq:preGSRepbeta} using \eqref{eq:preGSRbeta1} in $(0,\infty)\setminus (1-\epsilon, 1+\epsilon)$. 
	
	For the lower bound, we are left to discuss the compact interval $[1-\epsilon, 1+\epsilon]$. On this interval, we clearly have $(\beta +1)^{p}\asymp 1$.	The equivalence \eqref{eq:preGSRbeta1} shows in particular that the corresponding left-hand side is positive. Thus, we are left to show that there exists $C_{p,\epsilon}>0$ such that
	\[ g_t \geq C_{p,\epsilon}t\qquad \text{on }[1-\epsilon, 1+\epsilon].\]

	 \textbf{2.3.1. Lower bound for $1<p<2$ and $\beta=-\alpha \geq 0$:} 
	 By the discussion before we only have to show that $g_t \geq C_{p,\epsilon}t$ on $[1-\epsilon, 1+\epsilon]$. %This will be done in the spirit of the proof in \cite[Lemma~3.3]{AM}. 
	 Since $g_{t}' \geq 0$, we have for all $\beta \in [1-\epsilon, 1+\epsilon]$,
	 \[g_{t}(\beta)\geq g_{t}(1-\epsilon)= \frac{1}{1-t}\bigl(((1-\epsilon)(1-t)+t)^{p} - \abs{(1-\epsilon)(1-t)-t}^{p}  \bigr).\]
	 Using Lindqvist's lemma, Lemma~\ref{lem:Lindqvist90}, we get with $a=(1-\epsilon)(1-t)+t$ and $b=\abs{(1-\epsilon)(1-t)-t}$ that
	\begin{multline}\label{eq:preGSRfromLind}
		\abs{a}^p- \abs{b}^p\geq p\abs{(1-\epsilon)(1-t)-t}^{p-1}\bigl( (1-\epsilon)(1-t)+t-\abs{(1-\epsilon)(1-t)-t}\bigr) \\
		 + C_p\frac{\bigl((1-\epsilon)(1-t)+t-\abs{(1-\epsilon)(1-t)-t}\bigr)^2 }{\bigl((1-\epsilon)(1-t)+t+\abs{(1-\epsilon)(1-t)-t}\bigr)^{2-p}}.
	\end{multline}
If $(1-\epsilon)(1-t)-t\geq 0$, i.e., $t\in (0,(1-\epsilon)/(2-\epsilon))$, the latter reduces to
\begin{align*}
	\ldots&=  p\bigl((1-\epsilon)(1-t)-t\bigr)^{p-1}( 2t) + C_p\frac{(2t)^2 }{\bigl(2(1-\epsilon)(1-t)\bigr)^{2-p}}\\
	&=t \left( 2p\bigl((1-\epsilon)(1-t)-t\bigr)^{p-1} + 4C_p\frac{t }{\bigl(2(1-\epsilon)(1-t)\bigr)^{2-p}}\right)
\end{align*}
Using this, we get
\[g_t(\beta)\geq g_t(1-\epsilon)\geq t \left(2p\frac{\bigl((1-\epsilon)(1-t)-t\bigr)^{p-1}}{1-t}+ \frac{4C_p}{(2(1-\epsilon))^{2-p}}\cdot \frac{t}{(1-t)^{3-p}} \right).\]
Since $t\mapsto \bigl((1-\epsilon)(1-t)-t\bigr)^{p-1}/(1-t)$ is continuous on $[0,(1-\epsilon)/(2-\epsilon)]$, strictly positive on $[0,(1-\epsilon)/(2-\epsilon))$ and has a root at $t=(1-\epsilon)/(2-\epsilon)$, and $t\mapsto t/(1-t)^{3-p}$ is continuous and strictly positive on $(0,1)$, has a root at $t=0$, we conclude that there is a positive constant which bounds the sum from below on $[0,(1-\epsilon)/(2-\epsilon)]\subset [0,1]$.

If $(1-\epsilon)(1-t)-t<0$, i.e., $t\in ((1-\epsilon)/(2-\epsilon),1)$, then \eqref{eq:preGSRfromLind} reduces instead to
\begin{align*}
	\ldots&=  p\bigl(-(1-\epsilon)(1-t)+t\bigr)^{p-1}( 2(1-\epsilon)(1-t)) + C_p\frac{(2(1-\epsilon)(1-t))^2 }{(2t)^{2-p}}.
\end{align*}
Using this, we get
\[g_t(\beta)\geq g_t(1-\epsilon)\geq t \left(2p(1-\epsilon)\frac{\bigl(-(1-\epsilon)(1-t)+t\bigr)^{p-1}}{t}+ 2^pC_p(1-\epsilon)^2 \frac{1-t}{t^{3-p}} \right).\]
Since $t\mapsto \bigl(-(1-\epsilon)(1-t)+t\bigr)^{p-1}/t$ is continuous and strictly positive on $((1-\epsilon)/(2-\epsilon) ,1]$ and vanishes at $t=(1-\epsilon)/(2-\epsilon)$, and $t\mapsto (1-t)/t^{3-p}$ is continuous and strictly positive on $((1-\epsilon)/(2-\epsilon),1)$ (and is only zero at $t=1$), we conclude that there is a positive constant which bounds the sum from below.

This shows the desired lower bound for $1<p<2$ and $\beta \geq 0$.

	\textbf{2.3.2. Lower bound for $p>2$ and $\beta=-\alpha \geq 0$:} As in the case for $p<2$, it suffices to show that $g_{t}\geq t\cdot C_{p,\epsilon}> 0$ on $[1-\epsilon, 1+\epsilon]$. Since $g_{t}' \geq 0$, we have for all $\beta \in [1-\epsilon, 1+\epsilon]$,
	 \[g_{t}(\beta)\geq g_{t}(1-\epsilon)= \frac{1}{1-t}\bigl(((1-\epsilon)(1-t)+t)^{p} - \abs{(1-\epsilon)(1-t)-t}^{p}  \bigr).\]
	 Using Lindqvist's lemma, Lemma~\ref{lem:Lindqvist90}, we get with $a=(1-\epsilon)(1-t)+t$ and $b=\abs{(1-\epsilon)(1-t)-t}$ that
\begin{multline}\label{eq:preGSRfromLind2}
	\abs{a}^p- \abs{b}^p\geq p\abs{(1-\epsilon)(1-t)-t}^{p-1}\bigl( (1-\epsilon)(1-t)+t-\abs{(1-\epsilon)(1-t)-t}\bigr) \\
		+ C_p\abs{(1-\epsilon)(1-t)+t-\abs{(1-\epsilon)(1-t)-t}}^{p}.
	\end{multline}
If $(1-\epsilon)(1-t)-t\geq 0$, i.e., $t\in (0,(1-\epsilon)/(2-\epsilon))$, the latter reduces to
\begin{align*}
	\ldots&=  2tp\bigl((1-\epsilon)(1-t)-t\bigr)^{p-1} + C_p(2t)^{p}.
\end{align*}
Using this, we get
\[g_t(\beta)\geq g_t(1-\epsilon)\geq t \left(2p\frac{\bigl((1-\epsilon)(1-t)-t\bigr)^{p-1}}{1-t}+ 2^{p}C_p \frac{t^{p-1}}{1-t} \right).\]
Since $t\mapsto \bigl((1-\epsilon)(1-t)-t\bigr)^{p-1}/(1-t)$ is continuous on $[0,(1-\epsilon)/(2-\epsilon)]$, strictly positive on $[0,(1-\epsilon)/(2-\epsilon))$ and vanishes at $t=(1-\epsilon)/(2-\epsilon)$, and $t\mapsto t^{p-1}/(1-t)$ is continuous and strictly positive on $(0,1)$, and vanishes at $t=0$, we conclude that there is a positive constant which bounds the sum from below on $[0,(1-\epsilon)/(2-\epsilon)]$.

If $(1-\epsilon)(1-t)-t<0$, i.e., $t\in ((1-\epsilon)/(2-\epsilon),1)$, then \eqref{eq:preGSRfromLind2} reduces instead to
\begin{align*}
	\ldots&=  p\bigl(-(1-\epsilon)(1-t)+t\bigr)^{p-1}( 2(1-\epsilon)(1-t)) + 2^{p}C_p(1-\epsilon)^{p}(1-t)^{p}.
\end{align*}
Using this, we get
\[g_t(\beta)\geq g_t(1-\epsilon)\geq t \left(2p(1-\epsilon)\frac{\bigl(-(1-\epsilon)(1-t)+t\bigr)^{p-1}}{t}+ 2^pC_p(1-\epsilon)^p \frac{(1-t)^{p-1}}{t} \right).\]
Since $t\mapsto \bigl(-(1-\epsilon)(1-t)+t\bigr)^{p-1}/t$ is continuous and strictly positive on $((1-\epsilon)/(2-\epsilon) ,1]$ and vanishes only at $t=(1-\epsilon)/(2-\epsilon)$, and $t\mapsto (1-t)^{p-1}/t$ is continuous and strictly positive on $((1-\epsilon)/(2-\epsilon),1)$ and vanishes only at $t=1$, we conclude that there is a positive constant which bounds the sum from below.

This shows the desired lower bound for $p>2$ and $\beta \geq 0$, and we are left to show the upper bounds.

%%%%%%alte version mit eps =2/3
%%%	 \[g_{t}(\beta)\geq g_{t}(2/3)= \frac{(2/3)^p}{1-t}\bigl((1+t/2)^{p} - \abs{1-5t/2}^{p}  \bigr).\]
%%%	Using Lemma~\ref{lem:Lindqvist90}, we get with $a=(1+t/2)$ and $b=\abs{1-5t/2}$ that
%%%	\[\abs{a}^p- \abs{b}^p\geq p\abs{1-5t/2}^{p-1}\bigl( 1+t/2-\abs{1-5t/2} \bigr) + C_p\abs{1+t/2-\abs{1-5t/2}}^p.\]
%%%	If $0< t< 2/5$, this yields
%%%	\begin{align*}
%%%		(3/2)^pg_{t}(2/3)&\geq t\Bigl( 3 p\frac{(1-5t/2)^{p-1}}{1-t} + 3^p C_p\frac{t^p }{1-t}  \Bigr).
%%%	\end{align*}
%%%		Clearly, $(1-5t/2)^{p-1}/(1-t)$ is strictly decreasing with a root at $t=2/5$ and $t^p/(1-t)$ is strictly increasing with a root at $t=0$. Moreover, there exists $\tau \in (0, 2/5)$ depending on $p$, such that 
%%%	\[3 p\frac{(1-5\tau /2)^{p-1}}{1-\tau} = 3^p C_p\frac{\tau^p }{1-\tau}>0.\]
%%%	For this $\tau$, we have
%%%		\begin{align*}
%%%		(3/2)^pg_{t}(2/3)&\geq t\Bigl(  3^p C_p\frac{\tau^p }{1-\tau}  \Bigr)>0.
%%%	\end{align*}
%%%	
%%%	If $1> t > 2/5$, we have instead
%%%	\begin{align*}
%%%		(3/2)^pg_{t}(2/3)&\geq t\Bigl( 2 p\frac{(5t/2-1)^{p-1}}{t} + 2^p C_p\frac{(1-t)^{p-1} }{t}  \Bigr) .
%%%	\end{align*} With a similar argumentation, we get that there exists $\tau\in (2/5, 1)$ depending on $p$ such that
%%%	\begin{align*}
%%%		(3/2)^pg_{t}(2/3)&\geq t\Bigl( 2^p C_p\frac{(1-\tau)^{p-1} }{\tau}  \Bigr)> 0.
%%%	\end{align*}

	\textbf{2.3.3. Upper bound for $1<p<2$ and $p> 2$, and $\beta=-\alpha \geq 0$:} It remains to show that $g_t(\beta)\leq C_p t (\beta +1)^{p}$ for all $\beta \geq 0$.  
	
%	Recall that for all $p\geq 1$,
%	 \[g_{t}'(\beta)=p\bigl( \abs{\beta (1-t)+t}^{p-1}- \abs{\beta (1-t)-t}^{p-1}\sgn (\beta (1-t)-t)    \bigr)\]	
%	and
%	\[(t (\beta +1)^{p} )'= p t (\beta +1)^{p-1}. \]
	
	Recall that by the convexity of $\abs{\cdot}^{p}$, we have 
	\[\abs{a}^{p}- \abs{b}^{p}\leq p\abs{a}^{p-2}a (a-b), \qquad a,b\in \RR.\]
	
	Let $a=\beta (1-t)+t$ and $b=\abs{\beta (1-t)-t}$, then we get by the convexity that
	\[g_{t}(\beta)\leq p(\beta (1-t)+t)^{p-1}(\beta (1-t)+t-\abs{\beta (1-t)-t}).\]
	Since $\beta (1-t)+t\leq \beta +1$ and $(\beta +1)^{p-1}\leq (\beta +1)^{p}$ for all $\beta \geq 0$, $1<p<\infty$ and $t\in [0,1]$, we get
	\[\ldots \leq  p(\beta+ 1)^{p}(\beta (1-t)+t-\abs{\beta (1-t)-t}).\]
	If $\beta (1-t)\geq t$, then 
	$\beta (1-t)+t-\abs{\beta (1-t)-t}= 2t.$ 
	If $\beta (1-t)\leq t$, then 
	$\beta (1-t)+t-\abs{\beta (1-t)-t}= 2\beta (1-t)\leq 2t.$ 
	Thus, we get altogether,
	\[g_{t}(\beta)\leq 2pt(\beta+ 1)^{p}.\]	
	This finishes the proof of \eqref{eq:preGSRepbeta} and moreover, it also finishes the proof of \eqref{eq:preGSRep2}.
	
\emph{Ad~\eqref{eq:preGSRep1}:} The assertion follows by a simple case analysis. Here are the details: %for $a\in \set{0,t,1}$ the inequality clearly holds. Furthermore, l
Let \[f_{t,C}(a):=C\, t^{1/2}\abs{a-1}+ (1-t)\left(C\,\frac{\abs{a}+1}{2}-1\right)- \abs{a-t}, \qquad a\in \RR.\]
We have to show that $f_{t,C}\geq 0$ for all $t\in [0,1]$ and $C\geq 2$, $f_{t,C}\leq 0$ for all $t\in [0,1]$ and $0\leq C\leq 1/2$, and for every $C\in (1/2,2)$, the function $f_{\cdot,C}(\cdot)$ changes sign. 

\textbf{1. The cases $t\in \set{0, 1}$ and $a=t$:}
For $t=0$, we have
\[f_{0,C}(a)= \frac{C-2}{2}(\abs{a}+1),\]
which is non-negative for $C\geq 2$ and strictly negative for $C<2$.
If $t=1$, then
\[f_{1,C}(a)= (C-1)\abs{a-1},\]
which is non-negative for $C\geq 1$ and strictly negative for $C<1$. 
If $a=t$, then
\[f_{t,C}(t)= (1-t)\left( C \frac{2 t^{1/2}+t+1}{2}-1 \right),\]
which is non-negative for $C\geq 2$, non-positive for $C\leq 1/2$ and changes sign from negative to positive as $t$ increases in $1/2<C<2$. Hence, it is easy to see that $f_{\cdot,C}(\cdot)$ changes sign for any $1/2<C<2$ and an appropriate choice of $t$ by evaluating $f_{t,C}$ at $0, t$ and $1$.
 
\textbf{2. The remaining cases $t\in (0, 1)$, $a\neq t$:}
Note that for $a\not\in \set{0,t,1}$, we can calculate the derivative, i.e.,
\[f'_{t,C}(a)= C\,t^{1/2}\sgn(a-1)+ \frac{C}{2}(1-t)\sgn(a)-\sgn(a-t).\]
We have for all $t\in (0,1)$ and $C\geq 2$,
\begin{align*}
	f'_{t,C}(a)= 
	\begin{cases} 
	-C\,t^{1/2}+\frac{C}{2} t + \frac{2-C}{2}\leq 0, &\text{ for }   a< 0, \\
	-C\,t^{1/2}-\frac{C}{2} t+\frac{2+C}{2}\geq 0, &\text{ for }  0< a < t, \\	
	-C\,t^{1/2}-\frac{C}{2} t+\frac{C-2}{2}\leq 0, &\text{ for }  t< a <1, \\	
	\phantom{-}C\, t^{1/2}-\frac{C}{2} t+\frac{C-2}{2}\geq 0, &\text{ for } a > 1. \\
	\end{cases}
\end{align*}
If $0\leq C\leq 1/2$, then $f'_{t,C}$ has opposite sign on every subinterval.

Hence, $f_{t,C}$ has two extrema, one at $a=0$ and one at $a=t$. If $C\geq 2$, the extrema are minima, and if $0\leq C\leq 1/2$, the extrema are maxima. By the computations in the first case and since \[f_{t,C}(0)=C\,t^{1/2}-\frac{C}{2}t +\frac{C-2}{2},\] which is non-negative for $C\geq 2$ and non-positive for $0\leq C\leq 1/2$,  it follows that $f_{t,C}$ is non-negative if $C\geq 2$ and non-positive if $0\leq C\leq 1/2$ for all $t\in [0,1]$ and we have shown that the right-hand side in \eqref{eq:preGSRep1} is an upper bound for every $C\geq 2$, and lower bound if $0\leq C\leq 1/2$.

%Furthermore, we have for all $t\in (0,1)$ and $0\leq C\leq 1/2$,
%\begin{align*}
%	f'_{t,C}(a)= 
%	\begin{cases} 
%	-C\,t^{1/2}+\frac{C}{2} t + \frac{2-C}{2}\geq 0, &\text{ for }   a< 0, \\
%	-C\,t^{1/2}-\frac{C}{2} t+\frac{2+C}{2}\leq 0, &\text{ for }  0< a < t, \\	
%	-C\,t^{1/2}-\frac{C}{2} t+\frac{C-2}{2}\geq 0, &\text{ for }  t< a <1, \\	
%	\phantom{-}C\, t^{1/2}-\frac{C}{2} t+\frac{C-2}{2}\leq 0, &\text{ for } a > 1. \\
%	\end{cases}
%\end{align*}
%Hence, $f_{t,C}$ has two maxima, one at $a=0$ and one at $a=t$. By the computations in the first case and since \[f_{t,C}(0)=C\,t^{1/2}-\frac{C}{2}t +\frac{C-2}{2} \leq 0,\]  it follows that $f_{t,C}\geq 0$ for all $t\in [0,1]$ and we have shown that the right-hand side in \eqref{eq:preGSRep1} is a lower bound for every $0\leq C\leq 1/2$.

%The equivalence \eqref{eq:preGSRep2Neu} follows from applying \eqref{eq:preGSRep2} and \eqref{eq:preGSRep1}.

\emph{Ad~\eqref{eq:preGSRep3} and \eqref{eq:preGSRep4}:} We will show these inequalities similarly as we showed \eqref{eq:preGSRep2}. Recall that we have to show that
\begin{align*}
		t\abs{a-1}^2\leq t^{p/2}\abs{a-1}^p(\abs{a-t}+1-t)^{2-p},\qquad 1<p\leq 2,
	\end{align*}
	and
	\begin{align*}
		t\abs{a-1}^2(\abs{a-t}+1-t)^{p-2}\geq t^{p/2}\abs{a-1}^p, \qquad p\geq 2.
	\end{align*}
Note that the inequalities basically come from the fact that for $t\in [0,1]$, we have $t^{p/2}\geq t$ for $1<p\leq 2$, whereas $t^{p/2}\leq t$ for $p\geq 2$. Here are the details:

\textbf{1. The three cases $t\in \set{0, 1}$, $a=t$, and $p=2$:} If $p=2$, then it is obvious that we have equality for all $a\in \RR$ and $t\in [0,1]$.
	
	An easy computation shows that we indeed have equality for $t\in \set{0,1}$. 
	
	If $a=t$, then note that $t\in [0,1]$ implies $t^{p/2}\geq t$ for $1<p\leq 2$, and $t^{p/2}\leq t$ for $p\geq 2$. This immediately yields the desired inequalities.
	
	\textbf{2. The remaining cases $t\in (0, 1)$, $a\neq t$, and $p\neq 2$:} We consider the cases $a>t$ and $a<t$ separately.
	
	\textbf{2.1. The case $a > t$:} %We do the substitution $\alpha =(a-t)/(1-t)$. 
	Here, we have to show that
	  \begin{align}\label{eq:proofpregst3}
  	 t\abs{a-1}^2\leq  t^{p/2} \abs{a -1}^{p}(a +1-2t)^{2-p}, \qquad 1<p\leq 2
  \end{align}	
	as well as 
	  \begin{align}\label{eq:proofpregst4}
  	 t\abs{a-1}^2(a +1-2t)^{p-2}\geq  t^{p/2} \abs{a -1}^{p}, \qquad p\geq 2.
  \end{align} 	
  
	 Firstly, consider the case $1< p < 2$. We clearly have $a+1-2t \geq a-1$. Thus, $(a+1-2t)^{2-p}\geq (a-1)^{2-p}$. Moreover, $t\leq t^{p/2}$ for $t\in (0,1)$. This shows the inequality \eqref{eq:proofpregst3}.
  
  Secondly, consider the case $p> 2$. Because $(a+1-2t)^{p-2}\geq (a-1)^{p-2}$ as well as $t\geq t^{p/2}$ for $t\in (0,1)$, we get the desired inequality \eqref{eq:proofpregst4}.
  
  \textbf{2.2. The case $a < t$:} Note that $a<t<1$. Thus, we have to show that
	  \begin{align*}
  	 t(1-a)^2\leq  t^{p/2} (1-a)^{p}(1-a)^{2-p}, \qquad 1<p\leq 2
  \end{align*}	
	as well as 
	  \begin{align*}
  	 t(1-a)^2(1-a)^{p-2}\geq  t^{p/2} (1-a)^{p}, \qquad p\geq 2.
  \end{align*} 
	Since $t\leq t^{p/2}$ for $1<p<2$ and $t\geq t^{p/2}$ for $p>2$, we get the desired result.

\emph{Ad~\eqref{eq:preGSRep5}:} Recall that there we assume that $p\geq 0$. The desired inequality is clearly fulfilled if $\alpha =\beta =0$. Thus, assume that both do not vanish at the same time. Setting $t=\alpha/(\alpha +\beta)$, then \eqref{eq:preGSRep5} is equivalent to
\[f(t):=t^{p}+(1-t)^{p}\asymp 1, \qquad t\in [0,1].\]
If $0\leq p<1$, then $f$ has a minimum at $0$ and $1$, and a maximum at $1/2$. If $p\geq 1$, then $f$ has a maximum at $0$ and $1$, and a minimum at $1/2$. Since $f(0)= f(1)=1$ and $f(1/2)=2^{1-p}$, we finished the proof.
%
%The equivalence \eqref{eq:preGSRep5Neu} follows from applying \eqref{eq:preGSRep2Neu} and \eqref{eq:preGSRep5}.
%
%\underline{The situation for $a\in \set{\pm \infty}$}: Without loss of generality, let us only consider $a=\infty$. Let $(a_n)\in \RR$ be such that $\lim_{n\to\infty}a_n=\infty$. Then, for every $a_n$ the desired inequalities hold true as we have shown before. Then you can always take the limit of at least one side of the desired inequalities which will always either give you $0$ or $\infty$, and this does always make the inequality in the limit hold true.
\end{proof}

\subsection{Proof of the Ground State Representation Theorem}\label{sec:proofGSR}
\begin{proof}[Proof of Theorem~\ref{thm:GSR}]
Let $\phi \in C_c(V)$ and $0\leq u\in \FF(V)$ for some $V\sse X$. If either $u(x)=0$ or $u(y)=0$ for some $x,y\in V\cup \partial V$, then 
\[\abs{\nabla_{x,y}(u\phi)}^{p}-\p{\nabla_{x,y}u}\nabla_{x,y}(u\abs{\phi}^{p})=0.\]
Moreover, $u(x)u(y)(\nabla_{x,y}\phi)^{2}=0$. Thus, it remains in the following to consider the case $u(x), u(y)> 0$.

%Let us start with the first equivalence in \eqref{eq:GSRI}: 
Firstly, let $u(x)\geq u(y)>0$ for some fixed $x, y \in V\cup \partial V$.	Moreover, assume that $\phi(y)\neq 0$. Then, setting $t=u(y)/u(x)$ and $a=\phi(x)/\phi(y)$ in \eqref{eq:preGSRep2} combined with \eqref{eq:preGSRep1} results in
	\begin{align*}
		\abs{\nabla_{x,y}(u\phi)}^{p}&-(\nabla_{x,y}u)^{p-1}\nabla_{x,y}(u\abs{\phi}^{p})\\
%		&\asymp u(x)u(y)(\nabla_{x,y}\phi)^{2}\bigl( \abs{\nabla_{x,y}(u\phi)}+\abs{\phi(y)}\nabla_{x,y}u \bigr)^{p-2}.
&\asymp u(x)u(y)(\nabla_{x,y}\phi)^{2}\bigl( (u(x)u(y))^{1/2}\abs{\nabla_{x,y}\phi}+\frac{\abs{\phi(x)}+\abs{\phi(y)}}{2}\nabla_{x,y}u \bigr)^{p-2}.
	\end{align*}
	If $\phi(y)=0$, then we get the equivalence above if we can show that
	\[1-(1-t)^{p-1}\asymp t(t^{1/2}+(1-t)/2)^{p-2}, \qquad t=u(y)/u(x)\in (0,1].\]
	Using \eqref{eq:preGSRep1} with $a=0$, we see that $t^{1/2}+(1-t)/2\asymp 1$. Thus, it remains to show that
	\[1-(1-t)^{p-1}\asymp t, \qquad t\in (0,1].\]
	Now the latter follows easily from L'H\^{o}pital's rule (divide by $t$) and the fact  that the corresponding fraction is either increasing or decreasing.
	
	By a symmetry argument, we also get for all $x,y\in V\cup \partial V$ such that $u(y)\geq u(x)>0$,
	\begin{align*}
		\abs{\nabla_{x,y}(u\phi)}^{p}&-(\nabla_{y,x}u)^{p-1}\nabla_{y,x}(u\abs{\phi}^{p})\\
		&\asymp u(x)u(y)(\nabla_{x,y}\phi)^{2}\bigl((u(x)u(y))^{1/2}\abs{\nabla_{x,y}\phi}+\frac{\abs{\phi(x)}+\abs{\phi(y)}}{2}\nabla_{y,x}u\bigr)^{p-2}.
	\end{align*}
		
Note that by Green's formula, Lemma~\ref{lem:GreensFormula} for the $p$-Laplacian $L$,
\begin{align*}
\sum_{x,y\in V\cup\partial V}b(x,y)\p{\nabla_{x,y}u}\nabla_{x,y}(u\abs{\phi}^{p})&= 2\sum_{x\in V}u(x)L u(x)\abs{\phi(x)}^pm(x).
\end{align*}		
		
		Summing over all $x,y\in X$ with respect to $b$ and using the calculation above yields then in \eqref{eq:GSRI}. 	
%		
%		Now, we turn to the second equivalence in \eqref{eq:GSRI}: Firstly, let $u(x)\geq u(y)>0$ for some fixed $x, y \in X$.	Moreover, assume that $\phi(y)\neq 0$. Then, setting $t=u(y)/u(x)$ and $a=\phi(x)/\phi(y)$ in \eqref{eq:preGSRep1} results in
%\begin{align*}
%		 \abs{\nabla_{x,y}(u\phi)}&+ \abs{\phi(y)}\nabla_{x,y}u\\
%		 &\asymp (u(x)u(y))^{1/2}\abs{\nabla_{x,y}\phi}+\frac{1}{2}\left(\abs{\phi(x)}+\abs{\phi(y)}\right)\nabla_{x,y}u.
%	\end{align*}		
%	
%		If $\phi(y)=0$, then we get the inequalities above as well since for $t\in [0,1]$ we have $t^{1/2}+(1-t)/2\in [1/2,1]$, i.e., $t^{1/2}+(1-t)/2\asymp 1$. 
%
%By a symmetry argument, we also get the equivalence above for all $x,y\in X$ such that $u(y)\geq u(x)>0$. Summing over all $x,y\in X$ and using the first equivalence in \eqref{eq:GSRI} yields the second.
%
% Alternatively, we could have used \eqref{eq:preGSRep2Neu} instead of applying \eqref{eq:preGSRep2}. The proof can then be mimicked from the proof of the first equivalence.
\end{proof}

Now, we can directly continue with proving Corollary~\ref{cor:GSR}.
%\begin{proof}[Proof of Corollary~\ref{cor:GSR0}]	
%Let $\hh$ and $\HH$ denote the classical $p$-energy functional and classical $p$-Schrödinger operator, respectively. By \eqref{eq:GSRI}, we have
%		\[c_p h_{u}(\phi)\geq \hh(u\phi)- \ip{\HH u}{u \abs{\phi}^{p}}\geq \alpha(h(u\phi) - \ip{Hu}{u\abs{\phi}^{p}})+ \ip{\alpha Hu -\HH u}{u\abs{\phi}^{p}},\]
%		for some positive constant $c_p$. This yields the lower bound. The other estimate follows similarly. 
%		\end{proof}
		
\begin{proof}[Proof of Corollary~\ref{cor:GSR}]	
The proof of \eqref{eq:GSRI_p<2} and \eqref{eq:GSRI_p>2}  can simply be read of  \eqref{eq:GSRI}. 

The inequalities in \eqref{eq:GSRI_p>2Triangle} follow easily from   \eqref{eq:GSRI} and \eqref{eq:preGSRep5}.	
		\end{proof}
Alternatively, one can also deduce \eqref{eq:GSRI_p<2} and \eqref{eq:GSRI_p>2} from \eqref{eq:preGSRep4}, \eqref{eq:preGSRep3} and \eqref{eq:preGSRep2}. The proof can then be mimicked from the proof of Theorem~\ref{thm:GSR} and results in better constants.

\begin{remark}[Another Equivalence]
	Applying only \eqref{eq:preGSRep2} in the proof of Theorem~\ref{thm:GSR} yields the following equivalence: Let $0\leq u\in \FF(V)$, and $\phi\in C_c(V)$ for some $V\sse X$. We set for all $x,y\in X$
	\begin{align*}
	\phi_{u}(x,y):=\begin{cases}\phi(x),\quad \nabla_{x,y}u< 0, \\
	\phi(y), \quad \nabla_{x,y}u>0, \\
	\phantom{\phi()}0 ,\quad \nabla_{x,y}u=0.\end{cases}
	\end{align*}
Moreover, let the functional $h_u$ of $h$ with respect to $u$ on $C_c(V)$ be given by
\[h_{u,3}(\phi):=\sum_{x,y\in X}b(x,y) u(x)u(y)(\nabla_{x,y}\phi)^{2}\bigl( \abs{\nabla_{x,y}u\phi}+\abs{\phi_{u}(x,y)}\abs{\nabla_{x,y}u} \bigr)^{p-2},\]
where we again set $0\cdot \infty =0.$ %the summands on the right-hand side are understood to be zero if $1<p<2$ and $\abs{\nabla_{x,y}u\phi}+\abs{\phi_{u}(x,y)}\abs{\nabla_{x,y}u}=0$. 
 Then, we have 
\[ h(u\phi)- \ip{Hu}{u\abs{\phi}^{p}}\asymp h_{u,3}(\phi), \qquad \phi\in C_c(V),\]
with equality if $p=2$.
\end{remark}

\section{The Representation For Non-Local Operators in the Euclidean Space}\label{sec:frac}
%\begin{remark}[Fractional $p$-Schrödinger operators] \label{rem:frac} 
The statement for graphs can be transferred to non-local $p$-Schrödinger operators in the flavour of \cite{FS08}. This is because the main part of the proof of the representation comes from an elementary equivalence, which does not use any knowledge of the underlying space for the corresponding energy functional. A brief comparison is given in \cite[Remark~2.4]{FS08}. 

To be more specific, the corresponding statement for weighted non-local $p$-Schrödinger operators is as follows: Fix $d\geq 1$, $p>1$, a non-negative symmetric and measurable function $b$ on $\RR^{d}\times \RR^{d}$ and a measurable function $c$ on $\RR^{d}$. Set 
\begin{align*}
E(\phi)&:=\frac{1}{2}\iint_{\RR^{d}\times\RR^{d}}b(x,y)\abs{\nabla_{x,y}\phi}^{p}\dd x\, \dd y+\int_{\RR^{d}}c(x)\abs{\phi(x)}^{p}\dd x,\\
		E_{u}(\phi)&:=\iint_{\RR^{d}\times \RR^{d}}b(x,y) u(x)u(y)(\nabla_{x,y}\phi)^{2} \\
		&\qquad \cdot\left( \bigl(u(x)u(y)\bigr)^{1/2}\abs{\nabla_{x,y}\phi}+ \frac{\abs{\phi(x)}+ \abs{\phi(y)}}{2}\abs{\nabla_{x,y}u} \right)^{p-2}\dd x\, \dd y.
		%E_{u,1}(\phi)&:= \iint_{\RR^{d}\times \RR^{d}}b(x,y) (u(x)u(y))^{p/2}\abs{\nabla_{x,y}\phi}^{p}\dd x\, \dd y.\\
\end{align*}
In the following, technical assumptions are needed to circumvent a regularisation of principle value type, confer \cite[Section~2]{FS08}.

\begin{ass}\label{ass1}
Let $p>1$. Assume that there exists a family of symmetric and measurable functions $(b_{\epsilon})_{\epsilon >0}$, on $\RR^{d}\times \RR^{d}$ and a family of measurable functions $(c_{\epsilon})_{\epsilon >0}$, on $\RR^{d}$ such that $0\leq b_{\epsilon}\leq b$ on $\RR^{d}\times \RR^{d}$ and $b_{\epsilon}(x,y)\to b(x,y)$ for a.e. $x,y\in \RR^{d}$, respectively $c_{\epsilon}(x)\to c(x)$ for $x\in \RR^{d}$ as $\epsilon \to 0$.

Moreover, let $u$ be a positive, measurable function on $\RR^{d}$, such that the integrals
\[(H_{\epsilon}u)(x):= \int_{\RR^{d}}b_{\epsilon}(x,y)\p{\nabla_{x,y}u}\, \dd y+ c_{\epsilon}(x)u^{p-1}(x)\]
are absolutely convergent for a.e. $x\in \RR^{d}$, belong to $L_{1,\mathrm{loc}}(\RR^{d})$ and the limit $Hu:=\lim_{\epsilon \to 0}(H_{\epsilon}u)$ exists weakly in $L_{1,\mathrm{loc}}(\RR^{d})$.
\end{ass}
Now we state the main result of this section, the ground state representation formula.

\begin{theorem}[Ground state representation]\label{thm:fracGSR}
Let Assumption~\ref{ass1} be fulfilled. Then, whenever $E(u\phi), E_{u}(\phi),$ and $\int u(Hu)_{+}\abs{\phi}^{p}\dd x$ are finite, we have
\begin{align}\label{eq:fractional}
		E(u\phi)- \int_{\RR^{d}} u(x)Hu(x)\abs{\phi(x)}^{p}\dd x \asymp E_{u}(\phi), \qquad \phi\in C_c(\RR^d),
\end{align}
\end{theorem}

The proof of Theorem~\ref{thm:fracGSR} goes along the lines of the proof of \cite[Proposition~2.2]{FS08} doing the necessary changes.

\begin{proof}[Proof of Theorem~\ref{thm:fracGSR}] We can assume that $\phi$ is bounded. Otherwise, we replace $\phi$ by the function $\phi \min (1,n\abs{\phi^{-1}})$ and let $n\to \infty$ using monotone convergence.  Let denote by $L_{\epsilon}$ the Laplacian-part of $H_{\epsilon}$. By our assumptions, we have
\begin{multline*}
\iint_{\RR^{d}\times\RR^{d}}u(x)L_{\epsilon}u(x)\abs{\phi(x)}^{p}\dd x\,\dd y\\ =\frac{1}{2}\iint_{\RR^{d}\times\RR^{d}} b_{\epsilon}(x,y)\p{\nabla_{x,y}u} \nabla_{x,y}(\abs{\phi}^{p}u)\dd x\, \dd y.
\end{multline*}

	Now, the proof for $\epsilon >0$ is exactly as the proof of Theorem~\ref{thm:GSR}. Note that the constant in the equivalence comes from a pointwise equivalence and does only depend on $p$ (and not on $\epsilon$). By the assumptions, we have weak convergence for the integral containing $(H_{\epsilon}u)$ and dominated convergence for the remaining integral. Taking the limit yields \eqref{eq:fractional}.% and \eqref{eq:fractionalFS}.
\end{proof}
Set
\begin{align*}
		E_{u,1}(\phi):= \iint_{\RR^{d}\times \RR^{d}}b(x,y) (u(x)u(y))^{p/2}\abs{\nabla_{x,y}\phi}^{p}\dd x\, \dd y.
\end{align*}
 In analogy to Corollary~\ref{cor:GSR} we have the following result.
\begin{corollary}
Let Assumption~\ref{ass1} be fulfilled. Then, whenever $E(u\phi), E_{u}(\phi),$ $\int u(Hu)_{+}\abs{\phi}^{p}\dd x$, and also $E_{u,1}(\phi)$ are finite, we have for $p\geq 2$ and some positive constant $c_p$,  
\begin{align}\label{eq:fractionalFS}
		 E(u\phi)- \int_{\RR^{d}} u(x)Hu(x)\abs{\phi(x)}^{p}\dd x \geq c_p E_{u,1}(\phi), \qquad \phi\in C_c(\RR^d),
\end{align}
	and if $1<p< 2$, we get the reversed inequality.
\end{corollary}
\begin{proof}
	The proof is similar to the proof of Corollary~\ref{cor:GSR} and therefore omitted.
\end{proof}
	
In \cite{FS08}, only the estimate in \eqref{eq:fractionalFS} in the case of $p\geq 2$ is given. 

%\end{remark}
Note that choosing $b(x,y)=\abs{x-y}^{-d-ps}$, $0<s<1$, $x,y\in \RR^{d}$, results in the classical fractional $p$-Hardy inequality.

Moreover, note that for Sobolev-Bregman forms associated with the fractional $p$-Laplacians, another ground state alternative was found recently, see \cite{Bogdan}. %A similar inequality might also hold for graphs.
\section{Characterisations of Criticality}\label{sec:critical}
In this section, we will discuss the notion of criticality. For the history of this notion see \cite[Remark~2.7]{P07} or \cite[Section~5]{KePiPo1}. There it is stated that in the continuum the notion goes back to \cite{Si0} and was then generalised in \cite{Mur86, P88}. On locally summable weighted graphs, \cite{KePiPo1} is the first paper discussing criticality in the context of linear Schrödinger operators. See also \cite[Chapter~6]{KLW21} (and references therein) for corresponding results for linear Laplace-type operators on graphs.

Non-negative energy functionals associated with Schrödinger operators seem to divide naturally into two categories: the ones which are strictly positive, i.e., for which a Hardy inequality holds true, and the ones which are not strictly positive, i.e., for which the Hardy inequality does not hold. In the linear ($p=2$)-case, there are surprisingly many equivalent formulations to the statement that the Hardy inequality does (not) hold, for graphs confer \cite{KePiPo1}. For $p=2$, $c=0$ and $m=\deg$, this is exactly the division of graphs into transient and recurrent graphs.

Using our recently developed ground state representation, we will see that many of the characterisations in \cite{KePiPo1} remain characterisations also if $p\neq 2$. %The main result of this section is stated explicitly for \[B=C= \abs{\cdot}^{p-1}\sgn(\cdot).\] 
%A non-negative function $w\in C(X)$ gives rise to a ($p$-)form $w=w_p\colon C_c(X)\to [0,\infty)$ via
%\[w(\phi)=\sum_{x\in X}\abs{\phi(x)}^{p}w(x).\]
%We will change between the function $w$ and the corresponding form $w$ without further notice (as long as it is clear which representation we take).

Let $h$ be a functional which is non-negative on $C_c(V)$, $V\sse X$. Then, $h$ is called \emph{subcritical} in $V$ if the \emph{Hardy inequality} holds true in $V$, that is, there exists a positive function $w\in C(V)$ such that \[ h(\phi)\geq \ip{w}{\abs{\phi}^{p}}, \qquad \phi\in C_c(V).\]
%\Hmm{Hier sollte vielleicht noch das Maß mit rein, damit sich das später besser rauskürzt?}
	If such a positive $w$ does not exist, then $h$ is called \emph{critical} in $V$. Moreover, $h$ is called \emph{supercritical} in $V$ if $h$ is not non-negative on $C_c(V)$.
	
	Other names for a subcritical functional are sometimes strictly positive, coercive, or hyperbolic.
	
Before we can state the main result of this section, we need the following definition: A sequence $(e_n)$ in $C_c(V)$, $V\sse X$, of non-negative functions is called \emph{null-sequence} in $V$ if there exists $o\in V$ and $\alpha>0$ such that $e_n(o)=\alpha$ and $h(e_n)\to 0$.

Moreover, we define the \emph{variational capacity} of $h$ in $V\sse X$ at $x\in V$ via
\[\cc_{h}(x,V):=\inf_{\phi\in C_c(V), \phi(x)=1}h(\phi).\]

%Furthermore, it follows then from the ground state representation theorem, Theorem~\ref{thm:GSR}, that in the form sense
%\[h(\phi)- (muHu)(\phi/u) \asymp h_u(\phi/u), \qquad \phi\in C_c(X). \]

\begin{theorem}[Characterisations of criticality]\label{thm:critical}
	Let $p>1$. Furthermore, assume that there exists a positive superharmonic function in $X$. Then the following statements are equivalent:
	\begin{enumerate}[label=(\roman*)]
		\item\label{thm:critical1} $h$ is critical in $X$.
		%\item\label{thm:critical9} There does not exist a Green's function for any $x\in X$.
		\item\label{thm:critical2} For any $o\in X$ and $\alpha>0$ there is a null-sequence $(e_n)$ in $X$ such that $e_n(o)=\alpha$, $n\in\NN$.
		\item\label{thm:critical3} $\cc_{h}(x,X)=0$ for all $x\in X$.
%		\end{enumerate}
%		\begin{enumerate}[label=(\roman*), start=4]
		\item\label{thm:critical7} There exists a strictly positive harmonic function $u$  in $X$ such that $h_u$ is critical in $X$.% that is, if there exists a positive function $w$ such that $h_{u}(\phi)\geq \ip{w}{\abs{\phi}^{p}}$ for all  $\phi\in C_c(X)$, then $w=0$.
		\item\label{thm:critical8} For all positive harmonic functions $u$ in $X$, the ground state representation $h_u$ is critical in $X$.	
%	\end{enumerate}
%Moreover, if additionally $V=X$, then the statements above are also equivalent to:	
%	\begin{enumerate}[label=(\roman*), start=6]
		\item\label{thm:critical4Neu} For any positive superharmonic function $u\in \FF(X)$ in $X$ and any null-sequence $(e_n)$ in $X$ there exists a positive constant $\ccc$ such that %$0\leq e_n \leq u$ for all $n\in \NN$ and 
		$e_n(x)\to \ccc\, u(x)$ for all $x\in X$ as $n\to \infty$.
		\item\label{thm:critical4} There exists a strictly positive harmonic function $u\in \FF(X)$ in $X$ and a null-sequence $(e_n)$ in $X$  such that %$0\leq e_n \leq u$ for all $n\in \NN$ and 
		$e_n(x)\to u(x)$ for all $x\in X$ as $n\to \infty$. If $p\geq 2$, the sequence can be chosen such that  $0\leq e_n \leq u$ for all $n\in \NN$.	
	\end{enumerate}
	In particular, if one of the equivalent statements above is fulfilled, then there exists a unique positive superharmonic function in $X$ (up to linear dependence) and this function is strictly positive and harmonic in $X$.	
%%%%	Moreover, if $V$ is connected, then the words \emph{'strictly positive on $V$'} can be replaced by  \emph{'positive on $V\cup \partial V$'}.	
	%then this implies the following:
%	\begin{enumerate}[label=(\roman*), start=7]
%	\item\label{thm:critical5} There exists a unique positive superharmonic function (up to linear dependence) and this function is harmonic.
%	\end{enumerate}
\end{theorem}
%The statement \ref{thm:critical5} is actually equivalent to the others. This, is picked up in the follow-up paper \cite{F:AAP}.
%
We remark that in the continuum, the corresponding characterisation holds true on any subdomain of $\RR^{d}$, confer \cite[Theorem~4.15]{PP}. In Proposition~\ref{prop:fails}, we show that $h$ cannot be critical in $V$ whenever $V\subsetneq X$, assumed that a positive superharmonic function on $V$ exists.

We divide the proof of this main theorem into two subsections. In the first subsection, we show some more general auxiliary lemmata, and in the second subsection, we show the equivalences.

\subsection{Preliminaries}\label{sec:preliminaries}%Harnack Ineq und min-stable but not harmonic-stable
We start with a direct consequence of the ground state representation. %Recall that $\FF_{\cap}(V)=\FF_{b, \abs{\cdot}^{p-1}\sgn(\cdot)}(V)\cap \FF_{b, B}(V)$, and $\HH$ denotes the classical $p$-Schrödinger operator.
\begin{lemma}\label{lem:groundstate}
Let $p>1$ and $V\sse X$. Assume that there exists a function $u\in \FF(V)$ which is strictly positive in $V$ and superharmonic on $V$. Then, we have
\[h(\phi)\geq \ip{Hu}{u^{1-p}\abs{\phi}^{p}}, \qquad \phi\in C_c(V).\]
In particular, $h$ is non-negative on $C_c(V)$.
\end{lemma} 
\begin{proof}
By the ground state representation, Theorem~\ref{thm:GSR}, the statement follows easily since for some $c_p>0$,
\[h(\phi)- \ip{Hu}{u^{1-p}\abs{\phi}^{p}}\geq c_p h_{u}(\phi/u)\geq 0, \qquad \phi\in C_c(V). \qedhere\]
%Beweis ohne GSR:
%From Lemma~\ref{lem:Lemma41} it follows that %%\Hmm{Das ist Prä-Picone!}%%
%\begin{align*}
%\abs{\nabla_{x,y}(u\phi)}^p\geq \abs{\nabla_{x,y}u}^{p-2}(\nabla_{x,y}u)\bigl(\nabla_{x,y}(\abs{\phi}^pu)\bigr). %\\
%%&+c_p u^{p/2}(x)u^{p/2}(y)\abs{\phi(x)-\phi(y)}^p.
%\end{align*}
%Note that by Green's formula,
%\begin{align*}
%\sum_{x\in V}\sum_{y\sim x}b(x,y)&\abs{\nabla_{x,y}u}^{p-2}\bigl(\nabla_{x,y}u\bigr)\bigl(\nabla_{x,y}(\abs{\phi}^pu) \bigr) = 2\sum_{x\in V}\abs{\phi(x)}^pu(x)L u(x)m(x),
%\end{align*}
%where we used that $\phi=0$ on $\partial V$.
%
%By the assumption on $u$ we have for all $x\in V$
%\[L u(x)m(x)\geq \bigl(g(x)m(x)-c(x)\bigr)u^{p-1}(x),\]
%and that $u> 0$ on $V$. This implies that 
%	\begin{align*}
%	h(u\phi)&\geq \sum_{x\in V}\abs{\phi(x)}^pu(x)L u(x)m(x)+ \sum_{x\in V}c(x)\abs{(u\phi)(x)}^{p}\\
%	&\geq\sum_{x\in V}\abs{\phi(x)}^pu^{p}(x)g(x)m(x)\\
%	&\geq 0.
%	\end{align*}
%Setting $\psi= \phi / u$ on $V$, we get $h(\psi)\geq 0$ which yields the result.
\end{proof}
Note that the reversed statement in Lemma~\ref{lem:groundstate} is also true which is proven in a follow-up paper by the author \cite{F:AAP}. The corresponding statement is  known as an Agmon-Allegretto-Piepenbrink-type theorem  (see \cite{Allegretto, Piepenbrink, BLS} for a  linear version in the continuum, \cite{Do84,KePiPo1} for a  linear version in the discrete setting, \cite{PP} for a recent non-linear version in the continuum, and \cite{LenzStollmannVeselic} for a corresponding result on strongly local Dirichlet forms). %However, the proof of the negation involves the introduction of a hand full of methods on finite subsets and a detailed study of the limiting process. This would go far beyond the scope of this paper and is picked up in the follow-up paper by the author \cite{F:AAP}.

An immediate consequence of the definition of criticality is the following statement.
\begin{lemma}\label{lem:critical}
	Let $p>1$. Assume that there exists a strictly positive and superharmonic function $u\in \FF(X)$. Let $h$ be critical in  $X$. Then any strictly positive and superharmonic function $u\in \FF(X)$ is harmonic on $X$.
\end{lemma}
\begin{proof}
	Let $u$ be such a strictly positive superharmonic function. By Lemma~\ref{lem:groundstate} we have $h(\phi)\geq \ip{u^{1-p}Hu}{\abs{\phi}^{p}}$ for all $\phi\in C_c(X)$. Because $h$ is critical we get that $u^{1-p}Hu=0$ on $X$, and thus, $u$ is harmonic on $X$.
\end{proof}

Next we show that locally, i.e., on finite and connected sets, our graphs fulfil a so-called Harnack inequality for non-negative supersolutions. This inequality implies that non-negative supersolutions are either strictly positive or the zero function, and they do not tend to infinity in the interior of the graph. 

There is a long list of proofs of various Harnack-type inequalities for the $p$-Laplacian, see for instance for metric spaces \cite[Theorem~8.12]{Bjoern} where a $p$-Poincar\'{e} inequality is assumed. The corresponding analogue for linear Schrödinger operators on locally summable graphs can be found in \cite{KePiPo1}. The basic idea of the following proof of the local Harnack inequality can also be found in \cite{Prado}, where the standard $p$-Laplacian on locally finite graphs without potential (i.e., $c=0$) is considered, and \cite{HoS2}, where the standard $p$-Laplacian on finite graphs without potential, is considered.

\begin{lemma}[Local Harnack inequality]\label{thm:harnackIneq}
	Let $p>1$, $K\subseteq X$ be connected and finite, and $f\in C(X)$. Then there exists a positive constant $C_{K,H,f}$ such that for any $u\in \FF(K)$ which is non-negative on $K\cup \partial K$ and satisfies $Hu\geq fu^{p-1}$ on $K$, we have
	\[ \max_Ku \leq C_{K,H,f} \min_K u.\]	
	The constant $C_{K,H,f}$ can chosen to be monotonous in the sense that if $f\leq g \in C(X)$ then $C_{K,H,f}\geq C_{K,H,g}$. Specifically, for any $x\in K$ we have 
\begin{align*}
	u(y)\leq \biggl(\Bigl(\frac{\deg(x)+ c(x)- f(x)m(x)}{b(x,y)}\Bigr)^{\frac{1}{p-1}}+1\biggr)u(x),
\end{align*}
	for all $y\sim x$. In particular, we have $\deg+c- fm\geq 0$ on $K$. Furthermore, if $u(x)=0$ for some $x\in K$, then $u(x)=0$ for all $x\in K\cup \partial K$.
	
	Moreover, if $V\sse X$ is connected, then any function which is positive on $V\cup \partial V$ and superharmonic on $V$ is strictly positive on $V$.
\end{lemma}
%Note that the existence of such a non-negative function $u$ such that $Hu\geq f\abs{u}^{p-2}u$ on $K$ is always ensured by the zero function.
\begin{proof}
Let $K\subseteq X$ be finite and connected and let $u\in \FF(K)$ such that $u\geq 0$ on $X$ and $Hu\geq fu^{p-1}$ on $K$ for some $f\in C(X)$. 

If $u(x_0)=0$ for some $x_0\in K$, then we have
\begin{align*}
0= f(x_0)u^{p-1}(x_0)m(x_0)&\leq Hu(x_0)m(x_0)%&\leq \sum_{y\in X}b(x_0,y)\abs{u(y)}^{p-1}\sgn(-u(y))\\ 
\leq -\sum_{y\in X}b(x_0,y)u(y)^{p-1} \leq 0.
\end{align*} 
Thus, for all $x\sim x_0$, we have $u(x)=0$ and since $K$ is connected we infer by induction that $u(x)=0$ for all $x\in K\cup \partial K$. 

Hence, we can assume that $u>0$ on $K$. Because of $Hu\geq fu^{p-1}$, we have
\begin{align*}%\label{eq:Harnack1}
	\sum_{y\in X, \nabla_{x,y}u\leq 0}&b(x,y)\lvert \nabla_{x,y}u\rvert^{p-1} \\
	&\leq \sum_{y\in X, \nabla_{x,y}u>0}b(x,y)( \nabla_{x,y}u)^{p-1}+\bigl(c(x)- f(x)m(x)\bigr)u(x)^{p-1}
\end{align*}
for any $x\in K$. Furthermore, we deduce for every $x\in K$
\begin{align*}
	u^{p-1}(x)&\sum_{y\in X, \nabla_{x,y}u\leq 0}b(x,y)\Bigl\lvert 1-\frac{u(y)}{u(x)}\Bigr\rvert^{p-1} \\
	&\leq u^{p-1}(x)\biggl( \sum_{y\in X, \nabla_{x,y}u>0}b(x,y)\Bigl( 1-\frac{u(y)}{u(x)}\Bigr)^{p-1}+ c(x)- f(x)m(x)\biggr) \\
	&\leq u^{p-1}(x)\biggl( \sum_{y\in X, \nabla_{x,y}u>0}b(x,y)+c(x)- f(x)m(x)\biggr) \\
	&\leq u^{p-1}(x)\bigl( \deg(x)+ c(x)- f(x)m(x)\bigr).  
\end{align*}

Let $d_f=\deg+c- fm$. The previous calculation implies that $d_f\geq 0$ on $K$. Now assume that there is $y_0\sim x$ such that $u(x)\leq u(y_0)$. Then the previous calculation also implies
\begin{align*}
	b(x,y_0)\Bigl(\frac{u(y_0)}{u(x)}-1\Bigr)^{p-1}\leq d_f(x),\text{ i.e.,} \qquad u(y_0)\leq \biggl(\Bigl(\frac{d_f(x)}{b(x,y_0)}\Bigr)^{\frac{1}{p-1}}+1\biggr)u(x).
\end{align*}
Hence, for all $y_1\sim x$ we have
\begin{align*}
	u(y_1)\leq \biggl(\Bigl(\frac{d_f(x)}{b(x,y_1)}\Bigr)^{\frac{1}{p-1}}+1\biggr)u(x).
\end{align*}

Since $K$ is finite there exists a minimum and a maximum of $u$ in $K$. Let $x_{\min}$ and $x_{\max}$ denote the corresponding points, respectively. Moreover, let $x_0\sim x_1 \sim \ldots \sim x_n$ be a path in $K$ such that $x_0=x_{\min}$ and $x_n=x_{\max}$. Then we derive
\begin{align*}
u(x_0)\leq \prod_{i=0}^{n-1} \Biggl(\biggl(\frac{d_f(x_i)}{b(x_i,x_{i+1})}\biggr)^{\frac{1}{p-1}}+1\Biggr) u(x_n). 
\end{align*} 
Again since $K$ is finite there exists only a finite number of paths in $K$ between $x_{\min}$ and $x_{\max}$ with different vertices. Hence the minimum of the product on the right-hand-side exists. This minimum we denote by $C_{K,H,f}$. Clearly, $1\leq C_{K,H,f}< \infty$ since $K$ is finite and $u>0$. Moreover, if $f\leq g\in C(X)$ then $d_f\geq d_g$ and hence $C_{K,H,f}\geq C_{K,H,g}$.

We still have to show the last assertion. To do so, let $s\in \FF(V)$ be a positive superharmonic function on $V$. Then there exists $o\in X$ such that $s(o)>0$. Since $X$ is connected there exists a finite path from any $x\in X$ to $o$. Denoting this path by $K$, we can apply the first part of this lemma and get that $s(x)>0$. Since $x$ was arbitrary, we conclude that $s>0$ on $V$.
\end{proof}

The following lemma is the discrete analogue of \cite[Proposition~4.11]{PP}.
\begin{lemma}\label{lem:PP411}
	Let $p>1$ and assume that there exists a positive superharmonic function $u\in \FF$. Furthermore, assume that $(e_n)$ is a null-sequence in $X$ such that $e_n(o)=\alpha$ for some $o\in X$ and $\alpha > 0$. Then, $e_n\to (\alpha/u(o))u$ pointwise on $X$ as $n\to\infty$. In particular, for all $(x,y)\in X\times X$ we have $\nabla_{x,y}(e_n/u)\to 0$ as $n\to\infty$.
\end{lemma}
\begin{proof}
	By the Harnack inequality, Lemma~\ref{thm:harnackIneq}, any positive superharmonic function $u$ in $X$ is strictly positive in $X$. 
	 	
	Let $o\in X$ and $\alpha >0$ be arbitrary.  Set $\phi_n:=e_n/u$. Then, by the ground state representation, Theorem~\ref{thm:GSR},
	\[0\leq h_{u}(\phi_n)\asymp h(e_n)\to 0, \qquad n\to \infty.\]
	
Firstly, let $p\geq 2$. Then the equivalence implies $\abs{\nabla_{x,y}\phi_n}\to 0$ for all $x,y\in X$, $x\sim y$. Since $X$ is connected, we have for any $x\in X$ a $k\in \NN$ such that $x=x_1\sim \ldots \sim x_k=o$. Thus, we obtain
\begin{align*}
	\lim_{n\to\infty}\phi_n(x)= \lim_{n\to\infty} \bigl(\sum_{i=1}^{k-1} \nabla_{x_i,x_{i+1}}\phi_n + \phi_n(o)\bigr)= \alpha /u(o).
\end{align*}
Rearranging, yields $e_n\to (\alpha/u(o))u$ pointwise on $X$ as $n\to\infty$.
	
Secondly, let $1<p< 2$. Then the equivalence implies either
	\begin{enumerate}
	\item\label{1} $\abs{\nabla_{x,y}\phi_n}\to 0$, or
	\item\label{2} $\abs{\nabla_{x,y}\phi_n}\to \infty$, or
	\item\label{3} $(\phi_n(x)+\phi_n(y))\abs{\nabla_{x,y}u}\to \infty$ 
	\end{enumerate}
	   for each $(x,y)\in X\times X$, $x\sim y$. We show that \eqref{2} and \eqref{3} cannot apply: Using the triangle inequality, it is easy to see that \eqref{2} and \eqref{3} are equivalent for the pair $(x,o)$ with $x\sim o$. They are also equivalent to $e_n(x)\to \infty$ for $x\sim o$. Set
\[\Phi_n(x,o) :=(u(x)u(o))^{1/2}\abs{\nabla_{x,o}\phi_n}+ 1/2(\abs{\phi_n(x)}+\abs{\phi_n(o)})\abs{\nabla_{x,o}u}.\]	   
	   Then using Hölder's inequality with $\tilde{p}=2/p>1$, and $\tilde{q}=2/(2-p)$, we calculate
	   \begin{align*}
	   		&b(x,o)(u(x)u(o))^{p/2}\abs{\nabla_{x,o}\phi_n}^{p}\\
	   		&\leq \left( b(x,o)(u(x)u(o))\abs{\nabla_{x,o}\phi_n}^{2}\Phi^{p-2}_n(x,o)\right)^{p/2}\cdot \bigl(b(x,o)\Phi^{p}_n(x,o)\bigr)^{(2-p)/2}\\
	   		&\leq c_1(p)\cdot h^{p/2}_u(\phi_n)\\
	   		&\quad \cdot \left(b(x,o)\bigl((u(x)u(o))^{p/2}\abs{\nabla_{x,o}\phi_n}^{p}+c_2(p) (\abs{\phi_n(x)}^{p}+(\alpha /u(o)))\abs{\nabla_{x,o}u}^{p}\bigr)\right)^{(2-p)/2}\\
	   		&\leq c_1(p)\cdot h^{p/2}_u(\phi_n)\\
	   		&\quad \cdot \left(b(x,o)\bigl((u(x)u(o))^{p/2}\abs{\nabla_{x,o}\phi_n}^{p}+c_2(p) (\abs{\phi_n(x)}^{p}+(\alpha /u(o)))\abs{\nabla_{x,o}u}^{p}\bigr)+1\right)\\
	   		&\leq c_1(p)\cdot h^{p/2}_u(\phi_n) \cdot \left(b(x,o)\bigl(((u(x)u(o))^{p/2}+c_3(p))\abs{\nabla_{x,o}\phi_n}^{p}+c_4(p)\bigr)+1\right),
	   \end{align*}
	   where $c_i(p)$, $i\leq 4$, are positive constants depending only on $p$ (and not on $n$). Since $b(x,o)$, $u(x), u(o)$ are also independent of $n$ and strictly positive, we can rewrite the inequality above as
	\begin{align*}
	  \abs{\nabla_{x,o}\phi_n}^{p} \leq  C_1(p)\cdot h^{p/2}_u(\phi_n) \cdot \bigl(\abs{\nabla_{x,o}\phi_n}^{p}+C_2(p)\bigr),
	\end{align*}	   
	for some positive constants $C_i(p)$, $i=1,2$. Since $h_u(\phi_n)\to 0$ as $n\to\infty$, we conclude that $\abs{\nabla_{x,o}\phi_n}\to 0$, and $(e_n(x))$ does not converge to $\infty$ for all $x\sim o$. Hence, \eqref{2} and \eqref{3} cannot apply for all $x\sim o$, and only \eqref{1} holds true for all $x\sim o$. Thus, we can continue as in the case $p\geq 2$ to get that $e_n(x)\to (\alpha /u(o)) u(x)$ for all $x\sim o$.
	
	Arguing similarly, we have for all $y\sim x\sim o$ that
	\begin{align*}
	  \abs{\nabla_{y,x}\phi_n}^{p} 
	  \leq  C_1(p)\cdot h^{p/2}_u(\phi_n) \cdot \bigl(\abs{\nabla_{y,x}\phi_n}^{p}+C_2(p)\abs{\phi_n(x)}^p+C_3(p)\bigr) 
	\end{align*}	
	for some positive constants $C_i(p)$, $i\leq 3$. Thus, as before, \eqref{2} and \eqref{3} cannot apply for all $y\sim x\sim o$ which results in $e_n(y)\to (\alpha /u(o)) u(y)$. Since $X$ is connected, we get by induction that $e_n(y)\to (\alpha /u(o)) u(y)$ for all $y\in X$. This proofs the statement for $1<p<2$.
\end{proof}

\subsection{Proof of the Characterisations of Criticality}
	
Here, we proof Theorem~\ref{thm:critical}. We show the equivalences in the following order: \ref{thm:critical1} $\iff$ \ref{thm:critical2}	$\iff$ \ref{thm:critical3},
   and \ref{thm:critical1} $\iff$ \ref{thm:critical7} $\iff$ \ref{thm:critical8},
 and under the assumption $V=X$, \ref{thm:critical2} $\implies$ \ref{thm:critical4Neu}, 
 and (\ref{thm:critical1} \& \ref{thm:critical2} \& \ref{thm:critical4Neu})  $\implies$ \ref{thm:critical4} $\implies$ \ref{thm:critical1}. 
   From \ref{thm:critical4}, we deduce the last assertion of the theorem. %The last assertion is a consequence of the Harnack inequality.
	
\begin{proof}[Proof of Theorem~\ref{thm:critical}]
%	Ad \ref{thm:critical1} $\iff$ \ref{thm:critical9}: This is the negation of Theorem~\ref{thm:GreensFunction}.
%
%%%%Keep in mind that since there exists a superharmonic function, we get from Lemma~\ref{lem:groundstate} that $h\geq 0$ on $C_c(X)$.
	Ad \ref{thm:critical1} $\implies$ \ref{thm:critical2}: Let $w_n=1_o/n$ for $o\in X$ and $n\in \NN$. Then by the criticality of $h$ in $V$ we have the existence of a function $e_n\in C_c(X)$ such that $h(e_n)< \ip{w_n}{\abs{e_n}^{p}}$. By the reverse triangle inequality, we have $h(\abs{e_n})\leq h(e_n)$ and thus, we can assume that $e_n\geq 0$. By Lemma~\ref{lem:groundstate} we have that $h$ is non-negative in $C_c(X)$, and therefore we get
	\[0\leq h(e_n)< \ip{w_n}{\abs{e_n}^{p}}= e_n^p(o)m(o)/n.\]
	Hence, we can normalise $e_n$ such that $e_n(o)=\alpha$ for any $\alpha >0$. Altogether, $h(e_n)< \alpha^pm(o)/n$ and $(e_n)$ is a null sequence in $X$.
	
	Ad \ref{thm:critical2} $\implies$ \ref{thm:critical1}: Let $(e_n)$ be a null-sequence in $X$ with $e_n(o)=\alpha >0$ for some $o\in X$. Let $w\geq 0$ on $X$ such that $h(\phi)\geq \ip{w}{\abs{\phi}^{p}}$ for $\phi\in C_c(X)$. Then,
	\[ 0 = \lim_{n\to\infty}h(e_n)\geq \lim_{n\to\infty}\ip{w}{\abs{e_n}^{p}}\geq \lim_{n\to\infty}w(o)e_n^p(o)m(o)=w(o)\alpha^pm(o).\]
	Since $o\in X$ is arbitrary, $m(o)>0$ and $\alpha >0$, we get $w=0$ on $X$.
	
%	Ad \ref{thm:critical2} $\iff$ \ref{thm:critical3} with "for all": This follows immediately from the definitions. However, we saw in Section~\ref{sec:cap}, that the "for all" in \ref{thm:critical3} can be replaced by a "for some", see Proposition~\ref{prop:capKfinite}. 
Ad \ref{thm:critical2} $\iff$ \ref{thm:critical3}: This follows immediately from the definitions.

	Ad \ref{thm:critical1} $\iff$ \ref{thm:critical7} $\iff$ \ref{thm:critical8}: This follows from the ground state representation, Theorem~\ref{thm:GSR}. Note that the existence of such a strictly positive harmonic function is ensured by Lemma~\ref{lem:critical}.
	
	%Ad \ref{thm:critical1} \& \ref{thm:critical2} $\implies$ \ref{thm:critical4}:
	 
	 Ad \ref{thm:critical2}  $\implies$ \ref{thm:critical4Neu}: This is Lemma~\ref{lem:PP411}.
	
	Ad  (\ref{thm:critical1} \& \ref{thm:critical2} \& \ref{thm:critical4Neu})  $\implies$ \ref{thm:critical4}: The preamble ensures the existence of a positive superharmonic function $u$. By the the Harnack inequality, Lemma~\ref{thm:harnackIneq}, the function $u$ is a strictly positive superharmonic function in $X$.  By Lemma~\ref{lem:critical}, the criticality of $h$ in $X$ implies that any strictly positive superharmonic in $X$ is a strictly positive harmonic function in $X$.
	
	By \ref{thm:critical4Neu}, any null-sequence converges to a constant multiple of $u$. The existence of a null-sequence is ensured by \ref{thm:critical2}. This shows the first part.
	
	Let $p\geq 2$, and let $(e_n)$ be a null sequence such that $e_n(o)= u(o)$ for some $o\in X$, and $e_n\to u$. Consider the sequence $(e_n\wedge u)$, where $\wedge$ denotes the minimum. We show that it is a null-sequence. Indeed, since for all $\alpha, \beta, \gamma \in \RR$, we have
	\[\abs{\alpha\wedge \gamma - \beta \wedge \gamma}\leq \abs{\alpha - \beta},\]
	we conclude,
	\[0\leq h(e_n\wedge u)\asymp h_{u}(u^{-1}(e_n\wedge u))=h_{u}(u^{-1}e_n \wedge 1)\leq h_{u}(u^{-1}e_n)\asymp h(e_n).\]
	Thus,  $h(e_n\wedge u)\to 0$, and $e_n(o)=u(o)>0$, i.e., $(e_n\wedge u)$ is  a null sequence. Since $e_n\to u$, we conclude that $(e_n\wedge u)\to u$.

	Ad \ref{thm:critical4} $\implies$ \ref{thm:critical1}: By Lemma~\ref{lem:groundstate},  $0\leq h(e_n)\to 0$. Hence, $h$ is critical.		
	
	Thus, we have completed the proof of the equivalences. The last statement follows immediately from \ref{thm:critical4Neu}. %, and the last statement is a direct consequence of the Harnack inequality, Lemma~\ref{thm:harnackIneq}. 
	This finishes the proof.	
\end{proof}

Now, we show that $h$ cannot be critical on  any proper subset of $X$. 
\begin{proposition}\label{prop:fails}
	Let $V\subsetneq X$, and assume that there is a function $u\in\FF(V)$ which is positive in $V\cup \partial V$ and superharmonic in $V$. Then, $h$ is subcritical in $V$. 
	
	Furthermore, in every connected component $U$ of $V$ where $u$ does not vanish, there does not exists a null-sequence $(e_n)$ in $U$ with $e_n(x)=\alpha$ for $x\in \partial (X\setminus U)$ and some $\alpha >0$, and we have $\cc(x, U)>0$ for all $x\in \partial (X\setminus U)$.
\end{proposition}
\begin{proof}
	Since $u$ is superharmonic in $V$ and positive in $V\cup \partial V$, there is $o\in V\cup\partial V$ such that $u(o)>0$. By the Harnack inequality, Lemma~\ref{thm:harnackIneq}, we have that $u$ is strictly positive in the connected component $U$ of $V$, where $o$ is either on the boundary of $U$ or in $U$. Moreover, note that $\partial U \neq \emptyset$ since $V\subsetneq X$ and $X$ is connected. Without loss of generality, we can assume that $V\neq \emptyset$, and thus $U\neq \emptyset$. 
	
	Assume that $h$ is critical in $V$. Let $x_1\in U$ such that there exists $x\in \partial U$ with $x\sim x_1$, and set $w_n = 1_{x_1}/n$ for $n\in\NN$.  Then, by the criticality, Lemma~\ref{lem:groundstate} and the reversed triangle inequality, we have
	\begin{align}\label{eq:fails}
		0\leq h(\abs{\phi})\leq  h(\phi)< \abs{\phi(x_1)}^{p}m(x_1)/n, \qquad \phi\in C_c(V).
	\end{align}
	Using the ground state representation, Theorem~\ref{thm:GSR}, we get that $h(u\psi)\asymp h_{u}(\psi)$ for all $\psi\in C_c(V)$. Let $\phi_n:=u\cdot \psi_n\in C_c(U)$, where $\psi_n(x_1):=\alpha > 0$. Without loss of generality, assume that $0\leq \psi_n\in C_c(U)$. Then, $(\phi_n)$ is a null-sequence of $h$ in $U$, and $(\psi_n)$ is a null-sequence of $h_u$ in $U$.
	  
	  Firstly, let $p\geq 2$. Since $(\psi_n)$ is a null-sequence of $h_u$, we have $\abs{\nabla_{x,y}\psi_n}\to 0$ for all $x,y\in V\cup\partial V$, $x\sim y$. In particular, for $x_0\in \partial U$, we have 
	\[\abs{\psi_n(y)}=\abs{\nabla_{x_0,y}\psi_n}\to 0, \qquad y\in V, y\sim x_0.\]
	Thus, $\psi_n(y)\to 0$ for all $y\in U$ with $y\sim x_0$ since $u>0$ on $U$. But this is a contradiction, because $\psi_n(x_1)=\alpha > 0$ and $x_1\sim x_0$.
	
	Secondly, let $1<p<2$. Since $(\psi_n)$ is a null-sequence of $h_u$, we have for each $(x,y)\in (U\cup \partial U)^2$, $x\sim y$ either
	\begin{enumerate}
	\item\label{r1} $\abs{\nabla_{x,y}\psi_n}\to 0$, or
	\item\label{r2} $\abs{\nabla_{x,y}\psi_n}\to \infty$, or
	\item\label{r3} $(\psi_n(x)+\psi_n(y))\abs{\nabla_{x,y}u}\to \infty$.
	\end{enumerate}
	Since 
	\[\abs{\nabla_{x_0,x_1}\psi_n}=\alpha\in (0,\infty), \qquad x_0\in\partial U, x_0\sim x_1,\]
	we see that neither \eqref{r1} nor \eqref{r2} can apply for the pair $(x_0,x_1)$. Since
	\[(\psi_n(x_0)+\psi_n(x_1))\abs{\nabla_{x_0,x_1}u}=\alpha \abs{\nabla_{x_0,x_1}u}\in [0,\infty), \]
	also \eqref{r3} cannot apply. Thus we also have a contradiction in the case of $1<p<2$. 
	
	Thus, $(\phi_n)$ cannot be a null-sequence in $U$. Hence, the strict inequality in \eqref{eq:fails} does not hold, i.e., $h$ cannot be critical in $V$.
\end{proof}

\begin{remark}
 	Proposition~\ref{prop:fails} has the following interpretation for $p=2$, $m=\deg$, and $c= 0$: Given any connected graph, the induced graph on any proper subset is then a graph with boundary and thus transient.
 	 \end{remark}

We end this subsection by giving a connection between the energy functional associated with the graph $b$, and the energy functional associated with the graph $b_{u}$, where $b_{u}(x,y)=b(x,y)(u(x)u(y))^{p/2}$ for $0\leq u\in \FF$.

If $h$ is critical in $X$ then the unique harmonic function $\psi$ such that $\psi(o)=1$ is called \emph{(Agmon) ground state} of $h$ normalised at $o$.
\begin{corollary}
Let $p>1$, and $0\leq u\in \FF$.
\begin{enumerate}[label=(\alph*)]
\item\label{prop:51p>2} If $p>2$ and $u$ is a ground state of $h$, then $1$ is a ground state of $h_{u,1}$.
\item\label{prop:51p<2} If $1<p<2$ and $h_{u,1}$ is critical in $X$, then $h(\cdot)-\ip{u^{1-p}Hu}{\abs{\cdot}^{p}}$ is critical in $X$.
\end{enumerate}
\end{corollary}
\begin{proof}
	Ad~\ref{prop:51p>2}: Recall that a ground state is harmonic, i.e., $uHu=0$. Moreover, $h$ is critical. Then, by the ground state representation, Corollary~\ref{cor:GSR}, we have
that $h_{u,1}$ is critical. Since $1$ is a harmonic function with respect to the Laplace operator associated with $h_{u,1}$, it is a ground state.

	Ad~\ref{prop:51p<2}: This is a direct consequence of the ground state representation.
\end{proof}
\subsection{Liouville Comparison Principle}
Here, we show a consequence of the characterisations of criticality and the ground state representation which is usually referred to as a Liouville comparison principle, confer \cite[Section~11]{P07} and references therein for the linear case. For the counterpart in the continuum see \cite[Theorem~8.1]{PR15}, or \cite[Theorem~1.9]{PTT08}. %Note that due to our slightly different ground state representation in comparison to the continuum, the preamble of the following result differs also slightly from its analogue.

\begin{proposition}[Liouville comparison principle] 
Let $p>1$. Let $b$ and $\bb$ be two graphs on $X$, and $c, \ccc \in C(X)$ be two potentials. Let denote $h$ and $\hh$ the energy functionals with corresponding Schrödinger operators $H:=H_{b,,c,m,X,p}$ and $\HH:=H_{\bb,\ccc,m,X,p}$, respectively. Assume that the following assumptions hold true:
\begin{enumerate}[label=(\alph*)]
\item\label{cor:liou1} The energy functional $h$ is  critical in $X$ with ground state $u$.
\item\label{cor:liou2} There exists a positive $\HH$-superharmonic function, and also a positive $\HH$-subharmonic function $\uu$. 
\item\label{cor:liou3} There exists a constant $\alpha >0$, such that for all $x,y\in X$ we have \[b^{2/p}(x,y)u(x)u(y)\geq \alpha\, \bb^{2/p}(x,y)\uu(x)\uu(y).\]
\item\label{cor:liou4} There exists a constant $\beta >0$ such that for all $x,y\in X$ we have for $p\geq 2$, \[b^{1/p}(x,y)\abs{\nabla_{x,y}u}\geq \beta\, \bb^{1/p}(x,y)\abs{\nabla_{x,y}\uu},\]
and the reversed inequality holds for $1<p<2$.
%Wird nicht mehr gebraucht wegen der verbesserten GSR%\item\label{cor:liou5} There exists a constant $E>0$ such that for all $x,y\in X$ and $ \phi\in C_c(X)$, $0\leq \phi \leq u$, we have for all $p\geq 2$
%\[b(x,y)\abs{\nabla_{x,y}\phi}\geq E\,\bb(x,y)\abs{\nabla_{x,y}\frac{\uu}{u}\phi},\]
%and the reversed inequality holds for $1<p<2$.
%\item\label{cor:liou6} For all $x,y\in X$ with $\bb(x,y)>0$ if $p\geq 2$, and with $b(x,y)>0$ if $1<p<2$, we have \[\sgn (\nabla_{x,y}u)=\sgn (\nabla_{x,y}\uu).\]
\end{enumerate}
Then the energy functional $\hh$ is critical in $X$ with ground state $\uu$.
\end{proposition}
\begin{proof}
	By Theorem~\ref{thm:critical}, there exists a null-sequence $(e_n)$ with respect to $h$ such that $e_n\to u$ pointwise as $n\to \infty$. Denote  $\phi_n=e_n/u, n\in \NN$. From the ground state representation, Theorem~\ref{thm:GSR}, we get $h(e_n)\asymp h_{u}(\phi_n)$ for all $n\in \NN$. Hence, using \ref{cor:liou3} and \ref{cor:liou4}, we infer $h_{u}(\phi_n)\geq \gamma_1 \hh_{\uu}(\phi_n)$ for some constant $\gamma_1>0$. Now, \ref{cor:liou2} implies by Lemma~\ref{lem:groundstate} that $\hh$ is non-negative in $C_c(X)$. Using this fact, the calculation before, and the ground state representation,  we get for some constants $\gamma_2, \gamma_3>0$ that
	\[0\leq \hh(\uu \phi_n)\leq \gamma_2 \hh_{\uu}(\phi_n)\leq \gamma_3 h(e_n)\to 0,\qquad n\to\infty.\]
	Thus, $(\uu \phi_n)$ is a null-sequence for $\hh$ and by Theorem~\ref{thm:critical}, $\hh$ is critical in $X$. Since $\phi_n\to 1$, we get $\uu \phi_n\to \uu$, and by Theorem~\ref{thm:critical}, $\uu$ is the ground state.
\end{proof}
\noindent\textbf{Acknowledgements.} The author gratefully acknowledges inspiring discussions with Matthias Keller and Yehuda Pinchover.  Furthermore, the author thanks the Heinrich-Böll-Stiftung for the support.

\printbibliography
\end{document}